\documentclass{article}
\usepackage{graphicx} % Required for inserting images
\usepackage{amsthm,amssymb}
\usepackage[a4paper,top=1.5cm,bottom=1.5cm,left=1.5cm,right=1.5cm,marginparwidth=2cm]{geometry}
\usepackage{parskip}
\usepackage{todonotes}
\usepackage{caption}
\usepackage{subcaption}
\usepackage{hyperref}

\usepackage[auth-sc-lg]{authblk}
\usepackage{cleveref}

\newtheorem{theorem}{Theorem}[section]
\newtheorem{lemma}[theorem]{Lemma}
\newtheorem{corollary}[theorem]{Corollary}
\newtheorem{proposition}[theorem]{Proposition}

\theoremstyle{definition}
\newtheorem{definition}[theorem]{Definition}

\newcommand{\leo}[1]{#1}
\newcommand{\leon}[1]{#1}
\newcommand{\leonn}[1]{#1}

\newcommand{\vm}[1]{#1}
\newcommand{\kh}[1]{#1}
\newcommand{\mj}[1]{#1}
\newcommand{\rev}[1]{#1}
\newcommand{\revnew}[1]{#1}

\newcommand{\supp}[1]{\textsc{Supp}(#1)}
\newcommand{\semi}[1]{\overline{#1}}
\newcommand{\cC}{\mathcal{C}}

\title{\rev{When are quarnets sufficient to reconstruct semi-directed\\ phylogenetic networks?}\footnote{This paper received funding from the Netherlands Organisation for Scientific Research (NWO) under
projects OCENW.M.21.306 and OCENW.KLEIN.125. \rev{This material is partly based upon work supported by the National Science Foundation under Grant No. DMS-1929284 while the authors were in residence at the Institute for Computational and Experimental Research in Mathematics in Providence, RI, during the \emph{Theory, Methods, and Applications of Quantitative Phylogenomics} program.}}}
\author[2]{Katharina T. Huber}
\author[1]{Leo van Iersel\footnote{Corresponding author: \url{l.j.j.vaniersel@tudelft.nl}}}
\author[1]{Mark Jones}
\author[2]{Vincent Moulton}
\author[1]{Leonie Veenema - Nipius}
\affil[1]{Delft Institute of Applied Mathematics, Delft
    University of Technology, Mekelweg 4, Delft, 2628CD, The Netherlands}
            
\affil[2]{School of Computing Sciences, University of East Anglia, NR4 7TJ, Norwich, United Kingdom}

\begin{document}

\maketitle

\begin{abstract}
Phylogenetic networks are graphs that are used to represent evolutionary relationships between different taxa.
They generalize phylogenetic trees since for example, unlike trees, they permit lineages to combine.
Recently, there has been rising interest in \emph{semi-directed} phylogenetic networks, 
which are mixed graphs in which certain lineage combination events 
are represented by directed edges coming together, whereas the remaining edges are
left undirected. One reason to consider such networks is that it can be difficult to 
root a network using real data. In this paper, we consider the problem
of when a semi-directed phylogenetic network is defined or \emph{encoded} by 
the smaller networks that it induces on the 4-leaf subsets of its leaf set.
These smaller networks are called \emph{quarnets}.
We prove that semi-directed binary level-$2$ phylogenetic networks 
are encoded by their quarnets, but that this is not the case for level-$3$. In addition, we prove 
that the so-called \emph{blob tree} of a semi-directed binary network, a tree
that gives the coarse-grained structure of the network, is always encoded by the quarnets of the network. \rev{These results are relevant for proving the statistical consistency of programs that are currently being developed for reconstructing phylogenetic networks from practical data, such as the recently developed \textsc{Squirrel} software tool.}
\end{abstract}

\section{Introduction}

\emph{Phylogenetic networks} are graphs used to represent evolutionary relationships between different \emph{taxa} (e.g. species, languages or other evolving objects). They are a generalization of the well-known phylogenetic trees, which are restricted to representing tree-like evolution in which lineages split but cannot combine~\cite{bapteste2013networks}. Both unrooted, undirected as well as rooted, directed phylogenetic networks have been and are still being studied intensively~\cite{elworth2019advances,huson2010phylogenetic}. Recently, there has been rising interest in \emph{semi-directed} phylogenetic networks, which are unrooted and have undirected edges as well as directed edges \rev{(for an example, see Figure~\ref{fig:squirrel})}~\cite{allman2024identifiability,barton2022statistical,huebler2019constructing,linz2023exploring,martin2023algebraic,solis2017phylonetworks,wu2024ultrafast}. %Essentially, these networks are defined as \revnew{those networks that can be obtained from a directed phylogenetic network by forgetting the direction of all edges except for the ones that represent reticulations and suppressing the root.}
%in such a way that there is always some edge into which a root vertex can be inserted so that all of the undirected edges can then be directed so that a directed, acyclic phylogenetic network is obtained.
The reason that semi-directed networks have become more popular is that %it tends to be impractical to pin down
the location of the root of a network can often not be identified from real data~\cite{kong2022classes}. Even so, rather than reverting to completely undirected networks, semi-directed networks do permit
%a limited number of
\rev{directed edges} (called \emph{arcs}) that can be used to represent so-called \emph{reticulations}, in which two lineages combine into one lineage that is at the end of two arcs.
Such reticulations are commonly used to model reticulate evolutionary events such as hybridization, introgression, recombination or lateral gene transfer, and there are approaches that can be used to identify such events from real data (see e.g. \cite{solis2017phylonetworks}).
 \rev{For example, the taxon \emph{M.leucophaeus} in Figure~\ref{fig:squirrel} 
 is below two  %directed
 arcs which indicate\revnew{s} a reticulation event.} %represent a putative introgression event.} 
 \revnew{Essentially, semi-directed phylogenetic networks are defined as those networks that can be obtained from a directed phylogenetic network by forgetting the direction of all arcs, except for the arcs that represent reticulations, and suppressing the root.}

\begin{figure}
    \centering
    \includegraphics{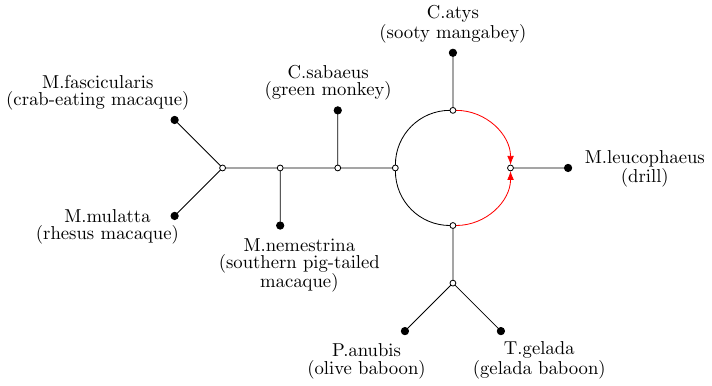}
    \caption{\rev{An example of a semi-directed phylogenetic network generated by the \textsc{Squirrel} software tool~\cite{holtgrefe2024squirrel} for an Old World monkey dataset~\cite{vanderpool2020primate} of Cercopithecinae. The edges are black
    and the arcs are red.}}
    \label{fig:squirrel}
\end{figure}

In this paper, we study the fundamental biological question of how much information is needed to reconstruct semi-directed phylogenetic networks, a question studied for rooted, directed phylogenetic networks in~\cite{huber2015much,van2022algorithm} and for unrooted phylogenetic networks in~\cite{erdHos2019not}. More concretely, we study which semi-directed evolutionary histories can be recovered from the evolutionary histories of groups of~$4$ taxa (called \emph{quarnets}). \rev{This is a topical issue since} several methods have been 
introduced recently to generate quarnets from DNA sequences or from gene trees \cite{allman2024nanuq+,cummings2023computing,holtgrefe2024squirrel,martin2023algebraic}. If a semi-directed phylogenetic network is uniquely determined by its induced subnetworks on sets of~$4$ taxa, then we say that the network is \emph{encoded} by quarnets.
\rev{Therefore, the question we study here can be formalized as the question of when a semi-directed phylogenetic network is encoded by its quarnets.}

%In this paper, we study which semi-directed phylogenetic networks are uniquely determined or {\em encoded}, by their 4-leaf induced subnetworks, or \emph{quarnets}.
This question is important for at least two reasons. The first reason is algorithmic. Accurate sequence-based phylogenetic network reconstruction methods (such as maximum likelihood) are often restricted to small numbers of taxa such as quartets. 
\rev{Hence, in order to prove that approaches which puzzle together 
quarnets into a larger semi-directed phylogenetic network are correct, we 
need to know when quarnets encode such networks.}
The other reason for studying quarnet encodings is that they can be used to show identifiability results for certain classes of phylogenetic networks from sequence data that is assumed to have evolved under some evolutionary model. In particular, the main idea is to prove identifiability of quarnets using techniques from algebraic geometry, and subsequently use quarnet encodings to generalize these results to larger networks~\cite{allman2022identifiability,ardiyansyah2021distinguishing,cummings2024invariants,gross2021distinguishing}.

\subsection{Previous results}

Encoding results for phylogenetic trees have been known for some time. Unrooted phylogenetic trees can be encoded by their splits, their quartets or by the distances between taxa~\cite{dress2012basic}. Similarly, rooted phylogenetic trees can be encoded by clusters, triplets or ultrametric distances. \revnew{Distances can still be used to identify some features of certain  networks~\cite{xu2023identifiability} and} some directed phylogenetic networks are still encoded by their triplets, which are 3-leaf trees %that are
contained in the network~\cite{gambette2012encodings,gambette2017challenge}. However, most networks are not encoded by their triplets. This led to research on binet, trinet and quarnet 
encodings~\cite{cardona2017reconstruction,huber2013encoding,van2014trinets,van2017binets}, which are 2-leaf, 3-leaf 
and 4-leaf subnetworks respectively, and can be either directed, undirected or semi-directed. 
Note that most of the results mentioned below are restricted to binary networks (whose 
internal non-root vertices have total degree~$3$).

General directed phylogenetic networks are not encoded by their trinets~\cite{huber2015much}. Hence, research has focused on encodings of subclasses of directed phylogenetic networks, e.g. by bounding their ``level''. A network is \emph{level-$k$} if it can be turned into a tree by deleting at most~$k$ edges/arcs from each blob. For example, networks~$N_d$ and~$N$ in \revnew{Figure~\ref{fig:degree2blob}} are level-2. Directed level-$1$ phylogenetic networks are encoded by their trinets~\cite{huber2013encoding}, and so are directed level-$2$ phylogenetic networks and other well-studied classes: so-called directed tree-child phylogenetic networks~\cite{van2014trinets} and directed orchard phylogenetic networks~\cite{semple2021trinets}. However, directed level-$3$ phylogenetic networks are not all encoded by their trinets~\cite{van2022algorithm}. On the algorithmic side, it has been shown that directed level-2 and orchard phylogenetic networks can be reconstructed from all their trinets in polynomial time~\cite{semple2021trinets,van2022algorithm}. For directed level-1 phylogenetic networks this is also possible and, moreover, a heuristic algorithm exists that constructs directed level-1 phylogenetic networks from practical data~\cite{oldman2016trilonet}. Encoding results have been used to show that this algorithm returns the correct network if its input data consists of all trinets of a directed level-$1$ phylogenetic network. Unfortunately, given any set of directed trinets (not necessarily one per triple of taxa) it is NP-hard to decide whether there exists a directed phylogenetic network that contains all given trinets, already for level-$1$~\cite{huber2017reconstructing}.

Much less is known about encodings for semi-directed phylogenetic networks. Two algorithms for constructing a semi-directed level-$1$ phylogenetic network from quarnets are given in~\cite{huebler2019constructing} but the paper does not prove explicitly that the algorithms always reconstruct the correct network, i.e. they do not prove that semi-directed level-$1$ phylogenetic networks are encoded by quarnets.
\revnew{Nevertheless, most features of level-$1$ phylogenetic networks are already determined by quartets ($4$-leaf trees contained in the network)~\cite{banos2019identifying}.}
\rev{Moreover, recently \textsc{Squirrel}~\cite{holtgrefe2024squirrel}, \revnew{NANUQ+~\cite{allman2024nanuq+}}, \textsc{Phynest}~\cite{kong2024inference}, \revnew{CUPNS~\cite{warnow2025advances}} and SNAQ~\cite{solis2017phylonetworks}
have been introduced for
generating level-$1$ semi-directed phylogenetic networks from quarnets, sequence alignments, \revnew{SNPs}
and collections of gene trees.}

\subsection{Our contribution}

In this paper, we study 
\mj{the quarnets of semi-directed phylogenetic networks.
%, i.e. their 4-leaf induced subnetworks. 
Reflecting the relative complexity of \rev{restricting a semi-directed network to a subset of its taxa},
%the definition of a semi-directed quarnet,
we show that \rev{this process} %a quarnet
%as our first contribution 
%formally define these subnetworks Section~\ref{sec:lev2}
% and to show in that section 
is well-defined (see Section~\ref{sec:lev2}).} While this is obvious for directed networks and level-1 semi-directed networks, for higher-level semi-directed networks it takes some care to prove that the intuitive definition works. \rev{Moreover,  
in our main result} we show that semi-directed binary level-$2$ phylogenetic networks are encoded by their quarnets:

%Note that 
%whilst \textsc{Squirrel} and other methods such as \textsc{Phynest} and \textsc{SNAQ} are limited to level-$1$, our results show the first steps towards developing similar methods that can handle with level-$2$. 

\textbf{Theorem~\ref{the:class-lev2}.}
\emph{The class of semi-directed, level-$2$, binary phylogenetic networks with at least~$4$ leaves is encoded by quarnets.}

\rev{Interestingly, this is the
theoretical limit for which semi-directed networks can be encoded, when categorizing networks by level.
More specifically,} we show that semi-directed level-$3$ phylogenetic networks are not \rev{all} encoded by their quarnets, which shows that there are fundamental 
limitations for extending methods to level-3 and higher:

\begin{theorem} The class of semi-directed, level-$3$, binary phylogenetic networks with at least~$4$ leaves is not encoded by quarnets.
\end{theorem}

The above theorem can be verified easily by considering the example in Figure~\ref{fig:lev3}, in which an example is presented of two different networks that have the same set of quarnets. Moreover, we note that the example can be extended to any number of leaves by inserting leaves between (or next to)~$a$ and~$b$ in~$N_1$ and in~$N_2$ (in any order).

In order to prove our main result (Theorem~\ref{the:class-lev2}) we show that 
the ``blob tree'' of a semi-directed phylogenetic network, also called the ``tree of blobs'', is uniquely determined by the quarnets of the network. Basically, a ``blob'' of a semi-directed network is a maximal subnetwork that cannot be disconnected by deleting a single edge/arc. The blob tree of such a network is obtained by contracting each blob to a single vertex (for more details, see Section~\ref{sec:cuts}). Blob trees have gained interest recently, since they represent the high-level branching structure of a network and may be identifiable even when the full network is not~\cite{allman2024tinnik,allman2023tree,van2014trinets,rhodes2024identifying}.
For all~$k\geq 1$, we show that the blob tree of a semi-directed binary level-$k$ phylogenetic network is always encoded by the quarnets of the network:

\begin{figure}
    \centerline{\includegraphics{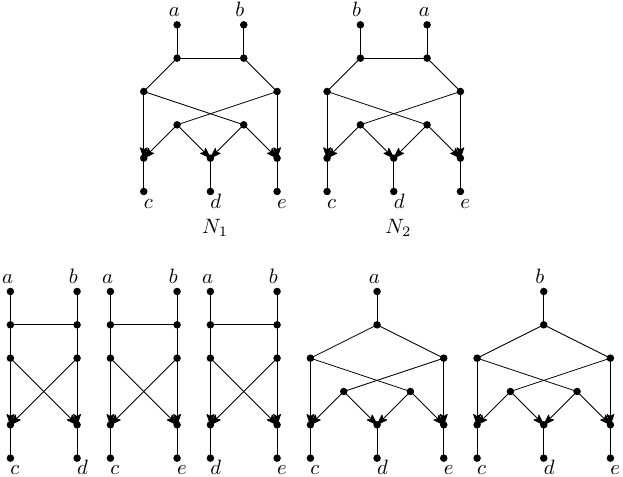}}
    \caption{\label{fig:lev3} Two semi-directed level-$3$ phylogenetic networks~$N_1$ and~$N_2$ (top) and their five quarnets (bottom). Even though~$N_1$ and~$N_2$ have exactly the same set of quarnets, the networks themselves are not isomorphic.}
    %Top: two non-isomorphic, simple, semi-directed level-$3$ phylogenetic networks~$N,N'$ with~$5$ leaves and with~$Q(N)\cong Q(N')$. Bottom: the five quarnets of~$N$ and~$N'$.
\end{figure}

\textbf{Corollary~\ref{cor:blobtree-encoding}.}
\emph{\rev{Suppose that~$N$ and~$N'$ are semi-directed phylogenetic networks on~$X$ with the same set of quarnets. Then~$N$ and~$N'$ have the same blob tree.}}

\rev{Note that this result was recently used in~\cite{holtgrefe2024squirrel} to prove that the
\textsc{Squirrel} program correctly reconstructs level-1 networks from perfect data.}

% \textbf{Corollary~\ref{cor:blobtree-encoding}.}
% \emph{Suppose that~$N$ and~$N'$ are semi-directed phylogenetic networks on~$X$ with~$Q(N)\simeq Q(N')$. Then $B(N)\cong B(N')$.}

This paper is based in part on preliminary results in the MSc thesis~\cite{nipius2022encoding}.

\subsection{Outline of the paper}

In \Cref{sec:prelims}, we give most of the main definitions used in this paper.
In \Cref{sec:restrictions}, we formally define the restriction of a (semi-)directed network to a subset of leaves and show it is well-defined. Based on this, we define quarnets and quarnet encodings in Section~\ref{sec:lev2}, where we also show that a semi-directed level-$k$ binary phylogenetic network with no non-trivial cut edges is encoded by its quarnets for $k \leq 2$. 
In \Cref{sec:cuts}, we show that the blob tree of any semi-directed level-$k$ binary phylogenetic network is encoded by its quarnets for all $k\geq 1$ or, equivalently, that the partition of the leaf set induced by a non-trivial cut edge is encoded by the quarnets. Combining the results from Sections~\ref{sec:lev2} and~\ref{sec:cuts}, in \Cref{sec:lev2combined} we show that semi-directed level-$2$ binary phylogenetic networks are encoded by their quarnets.
In \Cref{sec:discussion} we end with a discussion of possible future directions.

\section{Preliminaries}\label{sec:prelims}

%\todo[inline]{Comments Kathi: In preliminaries, It might be useful to start with the definition of a partially directed phy network on X and than  do directed binary phy network and undirected phy networks as a special case. For each one of them a phy tree would be again a special case. This might remove some repetition in the definition.

%In Section "suppression operations", it might be nice to have a figure illustrating them. Just for my sanity, is it possible that v is a reticulation in (V3) and (V4)? In (V5), is it possible that a network become disconnected? In (PAS) it might be better to replace "delete v" by "suppress v". 

%In Section 2.,2 it might be better to say that we are choosing a rooting $\rho$  first and then write something like  $N=S_{\rho}(N)$ as otherwise we might have a non-uniqueness problem. }

%\todo[inline]{Leo: We have changed the definitions below. In particular, we now consider four types of networks: directed, semi-directed, directed phylogenetic and semi-directed phylogenetic. These are all contained in the class of mixed graphs.

%The notion of a rooting has also changed. We now only consider rootings that have exactly one vertex more than the semi-directed network (namely the root).}

Let~$X$ be a finite set with~$|X|\geq 2$.

We consider mixed graphs which may have undirected edges and/or directed arcs and which may have parallel arcs. Undirected edges will simply be called \emph{edges} while directed edges will be called \emph{arcs}. When both are possible we will write ``edge/arc''. In this paper, there will be no reason to consider parallel edges or parallel edge-arc pairs. 
\revnew{Formally, a \emph{mixed graph} is an ordered tuple $G=(V,E)$ where~$V$ is a nonempty set of vertices,~$E$ is a multiset of undirected \emph{edges}~$\{u,v\}\subseteq V$, $u\neq v$, and directed \emph{arcs}~$(u,v)$ with~$u,v\in V$, $u\neq v$, such that each edge~$\{u,v\}$ has multiplicity at most~$1$ in~$E$ and such that for all arcs~$(u,v)\in E$ we have that $\{u,v\}\notin E$ and~$(v,u)\notin E$.}
A mixed graph is \emph{connected} if its underlying undirected graph contains a path between any two vertices.
The \emph{degree} of a vertex is the total number of incident edges and arcs. A \emph{leaf} is a degree-1 vertex. The \emph{indegree} of a vertex is the number of incoming arcs and the \emph{outdegree} is the number of outgoing arcs. A \emph{reticulation} is a vertex with indegree~$2$. Reticulations that are adjacent to a leaf are called \emph{leaf reticulations}.

For a set of vertices $S \subseteq V$ in a mixed graph $G=(V,E)$ with
vertex set~$V$ and edge/arc set~$E$, an edge/arc $e$ is \emph{incident} to $S$ if exactly one of its vertices is in $S$. If $e$ is an arc $(u,v)$ and $S\cap \{u,v\} = \{v\}$, we say $e$ is an arc \emph{entering} $S$ or an \emph{incoming} arc of $S$. If $S\cap \{u,v\} = \{u\}$, we say $e$ is an arc \emph{leaving} $S$ or an \emph{outgoing} arc of $S$.
%-- see Figure~\ref{fig:degree2blob} for an example.
% LEO: removed this line because the figure does not mention outgoing arcs and I think an example is not necessary
We also define $G[S]$
% := (S, \{e \in E: e \subseteq S\}$ % LEO: if $e$ is an arc then it is not a set
to be the subgraph of $G$ induced by $S$, i.e. the graph with vertex set~$S$, an edge~$\{u,v\}$ for each edge~$\{u,v\}$ in~$G$ with~$u,v\in S$ and an arc~$(u,v)$ for each arc~$(u,v)$ in~$G$ with~$u,v\in S$.

\subsection{Directed and semi-directed networks}

 \mj{Directed and semi-directed phylogenetic networks (defined formally below) are usually considered not to have parallel arcs or vertices of degree-$2$ (except for the root in directed phylogenetic networks). The restriction of a (directed or semi-directed) phylogenetic network to a subset of leaves is itself a  (directed or semi-directed) phylogenetic network. However, deriving the restriction involves the repeated application of reduction rules, some of which may result in mixed graphs with parallel arcs or degree-2 vertices. 
 For this reason, we consider a slight generalization of phylogenetic networks, simply  called \leon{(directed and semi-directed)} \emph{networks} (formally defined below), and reserve the qualifier \emph{phylogenetic} for a subclass of these graphs corresponding to the usual definition.}

Since we only consider binary networks in this paper, we do not include the word binary in the names of the network types defined below. \leon{We will include the word binary in the statements of theorems to avoid confusion.}

%{\color{blue} should we define cycles before for the following (see also comment below in section 2.4)?}
% LEO: I don't think that is necessary because I think we may assume that people know what a *directed* cycle is

\begin{definition}
A \emph{directed network} on~$X$ is a mixed graph~$N_d$, which may have parallel arcs, with the following restrictions:
\begin{itemize}
    \item $N_d$ has no undirected edges;
    \item $N_d$ has no directed cycles;
    \item each vertex has degree at most~$3$, indegree at most~$2$ and outdegree at most~$2$;
    \item there is a unique vertex with indegree~$0$, which has outdegree~$2$ and is called the \emph{root}; and
    \item the vertices with outdegree-$0$ %are leaves 
    \revnew{have indegree-$1$}
    and are bijectively labelled by the elements from~$X$.
\end{itemize}
\end{definition}

%{\color{blue} in the following we talk about suppression using operations -- 
%are these examples of suppression operations defined later in section 3, or do we want to use slightly 
%different terminology to avoid confusion?}
% LEO: Indeed, we mention in Section 3 that the operations here are special cases of the operations there (here they are only applied to the root)

\begin{definition}\label{def:semiunder}
A \emph{semi-directed network} on~$X$ is a mixed graph~$N$ that can be obtained from a directed network~$N_d$ on~$X$ by replacing all arcs %by
\revnew{with}
edges except for arcs entering reticulations and subsequently suppressing the root~$\rho$ if one of the following operations is applicable:
%~$v=\rho$: % neither may be applicable if the root has outgoing parallel arcs
\begin{itemize}
    \item if~$\rho$ is a degree-$2$ vertex with incident edges $\{u,\rho\},\{\rho,w\}$, replace these two edges by the edge~$\{u,w\}$ and delete~$\rho$; and
    % Leo: added the restriction that $v$ is a degree-2 vertex for later
    \item if~$\rho$ is a degree-$2$ vertex with an incident edge $\{u,\rho\}$ and an incident arc $(\rho,w)$, replace this arc and edge by the arc $(u,w)$ and delete~$\rho$.
\end{itemize}
%If this is the case then 
\mj{We call~$N_d$ a \emph{rooting} of~$N$.
If $N_d$ is a rooting of $N$,  we call~$N$ the \emph{underlying semi-directed network} of~$N_d$ and we write~$N=\semi{N_d}$.}
\end{definition}

\mj{See \Cref{fig:underylingSDs} for examples of directed and semi-directed networks.}
We note that semi-directed networks can have more than one rooting \mj{(see for example \Cref{fig:SDedgedge,fig:SDedgearc}).}
% If~$N_d$ is a rooting of~$N$, then we call~$N$ the \emph{underlying semi-directed network} of~$N_d$ and we write~$N=\semi{N_d}$.
Observe that $\semi{N_d}$ is well-defined, and that if $N_{d1}$ and $N_{d2}$ are rootings of the same semi-directed network $N$ then $\semi{N_{d1}}= \semi{N_{d2}} = N$. \revnew{Also note that it is possible that neither of the two suppressing operations in Definition~\ref{def:semiunder} is applicable (see Figure~\ref{fig:underylingSDs}(d)).}
% \mj{See \Cref{fig:intro} for an example of a direted network~$N_d$ and its underlying semi-directed network~$N$.}

\begin{figure}[h]
\begin{centering}
\begin{subfigure}{0.5\textwidth}
    \centerline{\includegraphics[scale = 0.5]{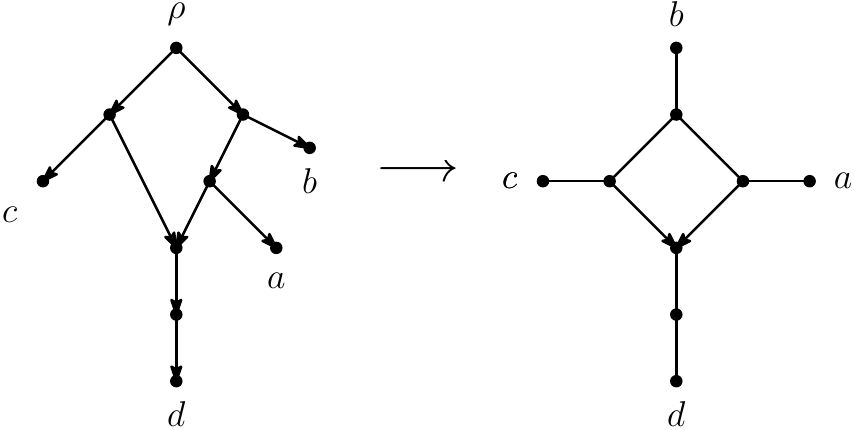}}
    \caption{\label{fig:SDedgedge}}
\end{subfigure}
\begin{subfigure}{0.5\textwidth}
    \centerline{\includegraphics[scale = 0.5]{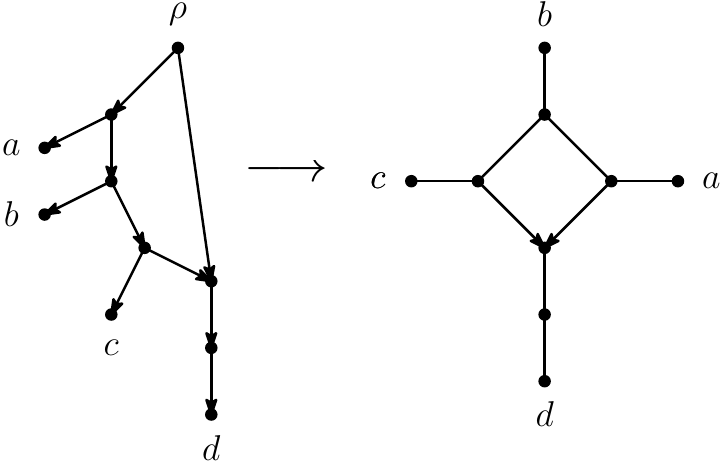}}  
    \caption{\label{fig:SDedgearc}}
\end{subfigure}
\begin{subfigure}{0.5\textwidth}
    \centerline{\includegraphics[scale = 0.5]{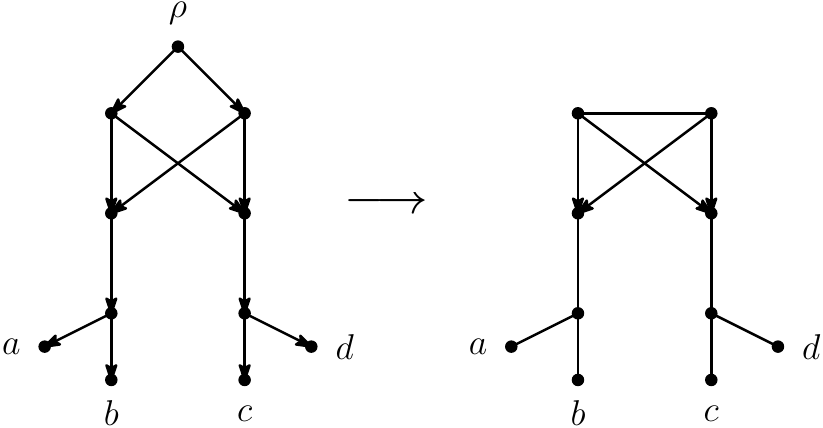}}
    \caption{\label{fig:SDblob}}
\end{subfigure}
\begin{subfigure}{0.5\textwidth}
    \centerline{\includegraphics[scale = 0.5]{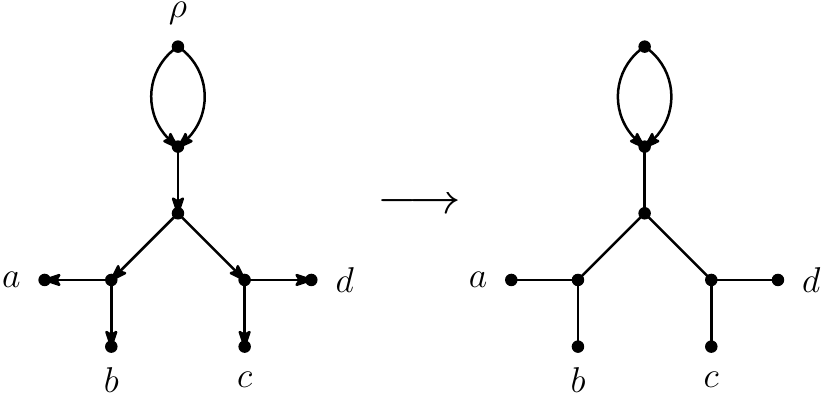}}
    \caption{\label{fig:SDparallel}}
\end{subfigure}
\begin{subfigure}{0.5\textwidth}
    \centerline{\includegraphics[scale = 0.5]{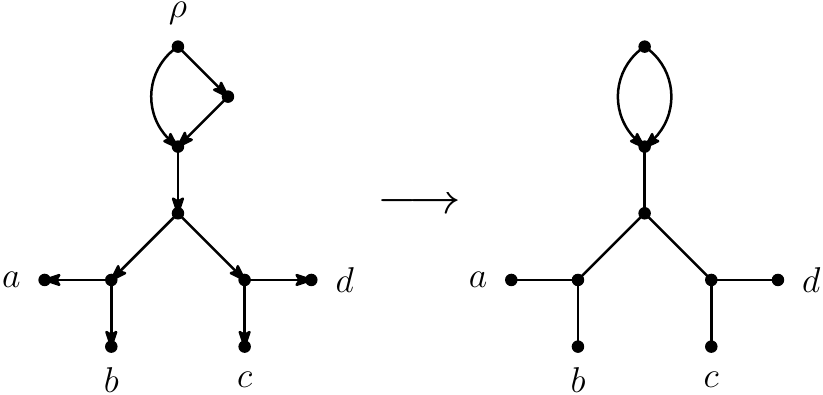}}
    \caption{\label{fig:SDtriangle}}
\end{subfigure}
\end{centering}
    \caption{\label{fig:underylingSDs} Some examples of a directed network (left) together with its underlying semi-directed network (right). Observe that the directed networks in \Cref{fig:SDedgedge} and \Cref{fig:SDedgearc} have the same underlying semi-directed network, as do the directed networks in  \Cref{fig:SDparallel} and \Cref{fig:SDtriangle}.}
\end{figure}

% We illustrate these definitions in Figure~\ref{fig:degree2blob}. \leon{Observe that, in this definition, one of (V1) and (V2) will be applicable unless there are parallel arcs leaving the root.}

\begin{figure}[h]
\begin{center}
\includegraphics{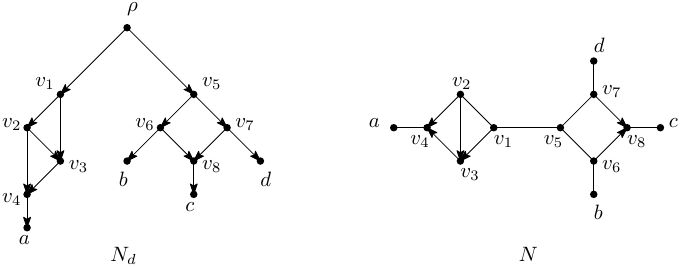}
\end{center}
\caption{\label{fig:degree2blob} \leon{A directed network $N_d$ and its underlying semi-directed network~$N$. Both~$N_d$ and~$N$ have two blobs, each having vertex sets~$\{v_1,\ldots ,v_4\}$ and~$\{v_5,\ldots ,v_8\}$. Note that neither~$N_d$ nor~$N$ is phylogenetic due to the blobs with vertex set~$\{v_1,\ldots,v_4\}$, \revnew{which have~$2$ incident edge/arcs}. Examples of $\wedge$-paths are~$(a,v_4,v_3,v_1,\rho,v_5,v_7,d)$ in~$N_d$ and~$(a,v_4,v_3,v_1,v_5,v_7,d)$ and~$(a,v_4,v_3,v_1,v_5,v_7,v_8,c)$ in~$N$. \revnew{An example of a cycle, in both~$N_d$ and in~$N$, is $(v_1,v_2,v_4,v_3,v_1)$ with sink~$v_4$.}}}
%Note that $N$ is not phylogenetic, as the blob $N[\rho, u_1,u_2,u_3,u_4]$ has two incident arcs.\todo[inline]{LEO: but I thought that the root blob may have only two outgoing arcs? Then the definition coincides with saying that contracting each blob to a single vertex gives a phylogenetic tree.}}
% \kh{A directed network $N$ on $X=\{a,b,c,d\}$ with a blob that has 2 outgoing arcs. (ii) A semi-directed network $N'$ on $X$ that is obtained from $N$. In both $N$ and $N'$, the subgraph induced by the highlighted vertices is a blob $B$. The arcs of $N$ (edges of $N'$) that share precisely one vertex with $V(B)$ are incident with $V(B)$ and are therefore incident with $B$.
% Note that $N$ is also phylogenetic because  a blob in the form of $B$ is allowed in a directed phylogenetic network. However, $N'$ is not phylogenetic because a blob with at most 2 incident edges is not allowed in a semi-directed network.  -- \bf{add leaf labels to the network on the left; highlight the vertices of the blob; indicate a $\Lambda$-path for both networks; add labels to interior vertices so that we can use figure as an illustration for the definitions surrounding cycles!}}
\end{figure}

We also note that a semi-directed network $N$ may have %a vertex $v$ with two outgoing
parallel arcs. This is the case if the directed network $N_d$ from which $N$ is obtained 
\leon{has parallel arcs or has its root in a triangle (i.e., an undirected length-3 cycle)}, \mj{as in \Cref{fig:SDparallel} and \Cref{fig:SDtriangle}}.
%as described in the definition is rooted at $v$ and $v$ has two parallel arcs leaving it in $N_d$ since, in this case, $v$ cannot be suppressed with operation (V1) or (V2).
% Leo: we don't suppress any pair of parallel arcs at this point.

A \emph{network} is either a directed or a semi-directed network.

% We note that networks are allowed to have vertices of degree 2, as well as parallel arcs. This is because we wish to preserve the property that any suppression operation applied to a network produces another network. We later define the concept of a \emph{phylogenetic} networks to explicitly exclude degree-2 vertices and parallel arcs.

A \emph{blob} of a mixed graph is a connected subgraph with at least three vertices  that is maximal under the property that deleting any edge/arc from the subgraph does not disconnect the graph 
\kh{-- see Figure~\ref{fig:degree2blob} for an example.}
An edge/arc~$e$ is \emph{incident} to a blob~$B$ if $e$ is incident with $V(B)$, the vertex set of~$B$.

\begin{definition}
A %(semi-directed or directed) Leo: deleted this because a network is by definition semi-directed or directed
network on~$X$ is called \emph{phylogenetic} if
\begin{itemize}
    \item it has no parallel arcs;
    \item it has no degree-2 vertices \kh{other than the root in case the network is directed}; and 
    %\todo{what about the root of a directed network?}
    \item it has no blobs with at most~$2$ incident edges/arcs, other than possibly a blob with no incoming and two outgoing arcs in case the network is directed.
    %\leon{unless the network is directed and the two incident arcs are leaving the blob}.
\end{itemize}
See Figure~\ref{fig:degree2blob} for an example \revnew{of how a network can fail to be phylogenetic}.
\end{definition}

A (directed/semi-directed) phylogenetic network with no reticulations is called a \emph{(rooted/unrooted) phylogenetic tree}.

%Observe that a directed (respectively semi-directed) network is phylogenetic precisely if it has no parallel arcs and contracting each blob into a single vertex gives a directed (respectively undirected) phylogenetic tree. The tree obtained in this way is called the \emph{blob tree}~$B(N)$ of a network~$N$.

We note that semi-directed phylogenetic networks as defined here do not contain any parallel arcs, even though some previous papers do allow one or more pairs of parallel arcs to be contained in such networks.
% EXAMPLE?

% not sure if we need the following two definitions

%An \emph{undirected binary phylogenetic tree} on~$X$ is an undirected tree having only vertices of degree~1 (\emph{leaves}) and degree~3 and with leaf set~$X$.

%An \emph{undirected binary phylogenetic network} on~$X$ is a  connected graph with no parallel edges and only vertices of degree~1 (\emph{leaves}) and degree~3, such that contracting each blob into a single vertex (one by one) gives an undirected phylogenetic tree on~$X$.

Two networks~$N,N'$ on~$X$ are \emph{isomorphic}, denoted $N\cong N'$, if there exists a bijection~$\phi$ from the vertex set of~$N$ to the vertex set of~$N'$ such that~$\{u,v\}$ is an edge of~$N$ if and only if $\{\phi(u),\phi(v)\}$ is an edge of~$N'$, $(u,v)$ is an arc of~$N$ if and only if $(\phi(u),\phi(v))$ is an arc of~$N'$ and $\phi(x)=x$ for all~$x\in X$. For sets of networks~$\mathcal{N},\mathcal{N'}$ \kh{on $X$,} we write $\mathcal{N}\simeq \mathcal{N'}$ if there exists a bijection $\psi:\mathcal{N}\to \mathcal{N'}$ such that $N \cong \psi(N)$ for all~$N\in\mathcal{N}$.

\subsection{Paths and cycles}

%{\color{blue} maybe move some of these definitions of cycles to the start of the section since we use them in Definition 2.1?}
% LEO: we only mention directed cycles in Def. 2.1

\leon{A \emph{path} in a network is a sequence of pairwise distinct vertices~$(v_1,\ldots,v_p)$, $p \ge 1$, such that
%if $p > 1$, then
for all $i\in\{1,\ldots ,p-1\}$ either $(v_i,v_{i+1})$ or $(v_{i+1},v_{i})$ is an arc or $\{v_i,v_{i+1}\}$ is an edge. Such a sequence is a \emph{semi-directed path} (from~$v_1$ to~$v_p$)} %in a network is a sequence of pairwise distinct vertices~$(v_1,\ldots,v_p)$ such that
if for all $i\in\{1,\ldots ,p-1\}$ either $(v_i,v_{i+1})$ is an arc or $\{v_i,v_{i+1}\}$ is an edge. Given two vertices~$u,v$ of a network, we say that~$v$ is \emph{below}~$u$ if there exists a semi-directed path from~$u$ to~$v$ (possibly $u=v$). If, in addition,~$u \neq v$ we say~$v$ is \emph{strictly below}~$u$. If $v$ is (strictly) below $u$ then we say $u$ is \emph{(strictly) above $v$}.

\revnew{We now introduce $\wedge$-paths, which can be pronounced as ``wedge paths''.}\footnote{\revnew{Such paths were called ``up-down paths'' in~\cite{xu2023identifiability}, but we use $\wedge$-paths to avoid confusion with an earlier notion of up-down paths that contain only arcs~\cite{bordewich2016determining,murakami2019reconstructing}.}} A \emph{$\wedge$-path} (between~$v_1$ and~$v_p$) in a network is a sequence of distinct vertices~$(v_1,\ldots,v_i,\ldots,v_p)$, $p \ge 1$, such that~$(v_i,\ldots ,v_1)$ and~$(v_i,\ldots ,v_p)$ are semi-directed paths, for some~$i\in\{1,\ldots ,p\}$ \kh{-- see Figure~\ref{fig:degree2blob} for an example.}
\mj{Such paths will be used when restricting a network to a subset of taxa.} %\leon{Observe that a path is a $\wedge$-path precisely if it does not contain a reticulation and both of its in-neighbours.}

A \emph{cycle} in a network~$N$ is a sequence $(v_1,e_1,v_2,e_2\ldots,v_p=v_1)$, $p \ge 4$, alternating between vertices $v_i$ and edges or arcs~$e_j$ such that~$v_i\neq v_j$ for $1\leq i<j<p$ and for all $i\in\{1,\ldots ,p-1\}$ either $e_i=(v_i,v_{i+1})$ or $e_i=(v_{i+1},v_{i})$ is an arc of~$N$ or $e_i=\{v_i,v_{i+1}\}$ is an edge of~$N$. \revnew{We may also describe a cycle by only its vertices~$(v_1,v_2,\ldots,v_p=v_1)$.}
%We say that two cycles in $N$ \emph{overlap} if they have at least one vertex in common. %\khn{In case}
%Since~$N$ is binary, two cycles in~$N$ overlap if and only if they have at least one edge or arc in common. In addition,
\revnew{We say that} a reticulation~$r$ in $N$ is a \emph{sink} of a cycle~$C$ if~$C$ contains both incoming arcs of~$r$. \revnew{See Figure~\ref{fig:degree2blob} for an example.}
%We call a cycle $C$ \emph{good} if it contains exactly one sink, and we call a good cycle \emph{excellent} if its sink is adjacent to a leaf. \mj{See Figure~\ref{fig:cyclesExamples}.}

A \emph{semi-directed cycle} in a network is a cycle~$(v_1,e_1,v_2,e_2\ldots,v_p=v_1)$ such that for all $i\in\{1,\ldots ,p-1\}$ either $e_i=(v_i,v_{i+1})$ or $e_i=\{v_i,v_{i+1}\}$. 

% \kh{For example, the network~$N$ in Figure~\ref{fig:degree2blob} has two cycles $C_1$ and $C_2$ and they overlap. For $i=1,2$, the vertex $r_i$ is a sink of $C_i$. Although both $C_1$ and $C_2$ are good, neither of them is excellent. Similarly, the network $N'$ in that figure contains cycles and these cycles overlap and are good but not excellent. Finally, $N$ does not contain a semi-directed cycle but both cycles in $N'$ are semi-directed.}

\begin{lemma}\label{lem:sdcycle}
In a \revnew{semi-directed %phylogenetic
network}~$N$ each cycle has at least one sink. In particular, $N$ contains no semi-directed cycles.
\end{lemma}
\begin{proof}
    Suppose~$N$ has a cycle~$C=\revnew{(v_1,v_2,\ldots,v_p=v_1)}$ without sinks. Let~$N_d$ be a rooting of~$N$. \revnew{Then $V(N_d) = V(N)\cup \{\rho\}$, with $\rho$ the root of $N_d$, and~$N_d$ either contains a cycle $(v_1,\ldots,v_p)$ or a cycle $(v_1,\ldots,v_{j-1},\rho,v_j,\ldots ,v_p)$.}

    \revnew{First suppose that~$N_d$ contains a cycle $(v_1,\ldots,v_p)$.
    Since~$N_d$ is acyclic,~$N_d$ contains some arc~$(v_{i-1},v_{i})$.  Following~$C$ from~$v_{i-1}$, at some point there is an arc~$(v_{k-1},v_k)$ followed by an arc~$(v_{k+1},v_k)$, again by the acyclicity of~$N_d$. However, then~$N$ also contains arcs~$(v_{k-1},v_k),(v_{k+1},v_k)$ and hence~$v_k$ is a sink of~$C$.}

    \revnew{Now consider the second case, that~$N_d$ contains a cycle $(v_1,\ldots,v_{j-1},\rho,v_j,\ldots ,v_p)$. Then we can conclude, similarly to the previous case, that~$N_d$ contains arcs $(u,v_k),(w,v_k)$ with~$u\in\{v_{k-1},\rho\}$ and $w\in\{v_{k+1},\rho\}$. In all cases,~$N$ contains arcs~$(v_{k-1},v_k),(v_{k+1},v_k)$ and hence~$v_k$ is a sink of~$C$.}
    
    \revnew{The second part of the lemma follows directly from the observation that a semi-directed cycle has no sink.}
\end{proof}

\section{\kh{Restricting networks}}\label{sec:restrictions}

\mj{In this section, we formally define the \emph{restriction $N|_A$} of a network $N$ on $X$ to a subset of taxa $A \subseteq X$ and consider some of its properties}.
In subsequent sections our focus will be on quarnets coming from a network, which are simply restrictions to subsets of size 4.

\mj{Roughly speaking, for a (phylogenetic) network $N$ on $X$ and a subset $A \subseteq X$, there are two main steps to constructing $N|_A$:
\begin{enumerate}
    \item Delete all vertices that are not contained on
    any path between two leaves in$~A$, resulting in a (not necessarily phylogenetic) network on $A$.
    \item Transform this network to a phylogenetic network on $A$ by repeatedly suppressing degree-2 vertices, parallel arcs, and blobs with at most~$2$ incident edge/arcs.
\end{enumerate}
In the remainder of this section, we make the above steps precise, and show that $N|_A$ is well-defined. 
The main technical task is to prove the intuitively obvious but non-trivial fact that \revnew{for the suppression operations described in step~2 the order does not affect} the final network, which implies that the restriction is well-defined.}

\subsection{Suppression operations}

\mj{We now formally define the suppression operations that are used to reduce a network to a phylogenetic network.} \leon{See Figure~\ref{fig:supp} for illustrations focusing on semi-directed networks.}

The \emph{blob suppression} operation on a network does the following for every blob~$B$ with at most two incident edges/arcs 
\leon{that are not two arcs leaving~$B$}:
\begin{itemize}
    \item[(BLS)] collapse~$B$ to a single vertex~$v_B$ and, if~$v_B$ has degree~$1$, delete it. 
\end{itemize}

The \emph{parallel arc suppression} operation on a network \kh{$N$} does the following for each pair of vertices~$u,v$ with two arcs $(u,v)$:
\begin{itemize}
\item[(PAS)] if $u$ and $v$ both have degree $3$ then remove the arcs $(u,v)$, replace \leon{any arc $(v,w)$ by $(u,w)$, any edge $\{v,w\}$ with $\{u,w\}$} and delete $v$.
%\todo{Leo: not sure if we need the degree-$3$ restriction}
    % \item[(PAS)] \kh{If $N$ is directed and $u$ is not the root of $N$ then} remove the arcs~$(u,v)$,  replace every arc $(v,w)$ with the arc $(u,w)$ and delete $v$. Then if~$u$ has degree~$1$, delete it.
    % \kh{If $N$ is semi-directed then remove the arcs~$(u,v)$, replace every edge 
    % $\{v,w\}$ with the edge $\{u,w\}$, and delete $v$. Then if~$u$ has degree~$1$, delete it. } 
    
    % \item[(PAS)] remove the arcs~$(u,v)$, collapse~$u$ and~$v$ into a single vertex~$w$ and, if~$w$ has degree~$1$, delete it.
\end{itemize}

\mj{The \emph{vertex suppression} operations on a network apply, for each degree-2 vertex~$v\in V$, one of the following if applicable}

\begin{itemize}
    \item[(V1)] if~$v$ has incident edges $\{u,v\},\{v,w\}$, replace them by an edge~$\{u,w\}$ and delete~$v$;
    \item[(V2)] if~$v$ has an incident edge $\{u,v\}$ and an incident arc $(v,w)$, replace them by an arc $(u,w)$ and delete~$v$; and
    %\item[(V1)] if~$v$ has incident edges $\{u,v\},\{v,w\}$, replace them by an edge~$\{u,w\}$ and delete~$v$;
    %\item[(V2)] if~$v$ has an incident edge $\{u,v\}$ and an incident arc $(v,w)$, replace them by an arc $(u,w)$ and delete~$v$; and
    \item[(V3)] if~$v$ has incident arcs $(u,v),(v,w)$, replace them by an arc $(u,w)$ and delete~$v$.
% the following does not occur in directed or semi-directed networks
%    \item[(V4)] if~$v$ has an incident arc $(u,v)$ and an incident edge $\{v,w\}$, replace them by an arc $(u,w)$ and delete~$v$.
\end{itemize}

\begin{figure}
    \centerline{\includegraphics{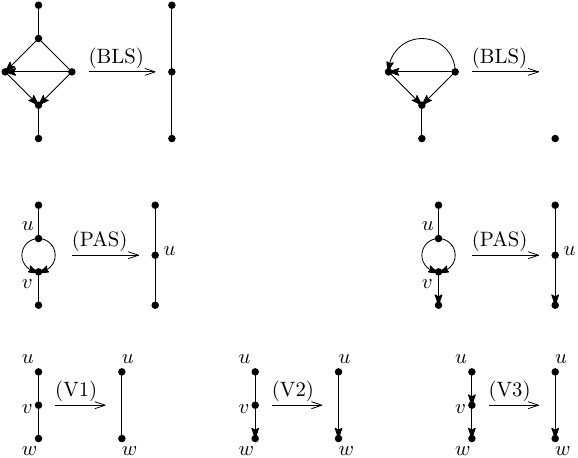}}
    \caption{\label{fig:supp} \leon{Illustrations of the suppression operations used to turn a semi-directed network into a semi-directed phylogenetic network.}}
\end{figure}

\revnew{Note that in a directed network only operation (V3) may be applicable.} \mj{Also observe that in the definition of a semi-directed network, (V1) and (V2) are %the rules that are
applied to the root $\rho$ (and only $\rho$) after replacing arcs not entering reticulations with edges.

\revnew{Note that parallel edges will never appear. To see this, recall from Lemma~\ref{lem:sdcycle} that each cycle in a semi-directed network has a sink and observe that this property is preserved under each of the suppression operations. Furthermore, a degree-$2$ vertex~$v$ with an incident arc~$(u,v)$ and edge~$\{v,w\}$ will never appear. To see this, observe that semi-directed networks have the property that, for each arc~$(u,v)$,~$v$ has indegree-$2$ and this property is preserved under each of the suppression operations.}}

It is easy to verify that if $N'$ is derived from $N$ by any of (V1), (V2), (V3), (BLS), (PAS) and $N$ is a \mj{directed} network, then $N'$ is a \mj{directed} network. The following lemma shows that this also holds for semi-directed networks, since (V3) is not applicable in semi-directed networks.

\begin{lemma}\label{lem:semiDirectedPreserved}
Let $N$ be a semi-directed network.
    If $N'$ is derived from $N$ by a single application of (V1), (V2), (BLS) or (PAS), then $N'$ is also a semi-directed network.
\end{lemma}

\mj{The proof of \Cref{lem:semiDirectedPreserved} is deferred to the appendix.}

The \emph{suppression} operation on a network~$N$ performs first the blob suppression operation \mj{(BLS)} and then repeatedly applies the parallel arc suppression operation \mj{(PAS)} and the vertex suppression operations \mj{(V1),(V2),(V3) until none of them} is applicable. The resulting network is denoted $\supp{N}$.

%Observe that a network~$N$ is phylogenetic if and only if $\supp{N}\cong N$ (i.e. no suppression operation is applicable).\todo{This only holds if we delete the degree-$3$ restriction from (PAS).}

\mj{The proof of the following result is quite technical, and is deferred to the appendix.}

\leonn{\begin{lemma}\label{lem:SuppWellDefinedBrief}
    $\supp{N}$ is well-defined for any network~$N$. %directed and semi-directed networks.
\end{lemma}}

\subsection{Restrictions}

Given a network~$N$ on~$X$ and a subset~$A\subseteq X$ with~$|A|\geq 2$, we define $N_{\wedge A}$ as the network obtained from~$N$ by deleting all vertices that are not on a $\wedge$-path between two vertices in~$A$. %\leon{(called a $\emph{\wedge A$-path} for short)}.
The \emph{restriction} of~$N$ to~$A$ is defined as $N|_A=\supp{N_{\wedge A}}$.
See Figure~\ref{fig:restr} for an example. \leon{Note that for a directed network~$N$ it is not true in general that $\semi{N|_A}\cong \semi{N}|_A$ since suppression operations may be applicable in $\semi{N|_A}$.} \revnew{Consider for example the directed network~$N$ in Figure~\ref{fig:underylingSDs}(c). Then $N|_{a,b,c,d}$ is equal to~$N$ and $\semi{N|_{a,b,c,d}}$ is the indicated semi-directed network. However, $\semi{N}|_{a,b,c,d}$ is an unrooted phylogenetic tree since the blob with two incident edges is suppressed.}

% \begin{figure}\label{fig:wedgeNetworks}
% \begin{center}
% \includegraphics[width = 0.5\textwidth]{wedgeD}

% \includegraphics[width = 0.5\textwidth]{wedgeSD}

% \includegraphics[width = 0.5\textwidth]{SuppA}
% \caption{1. Example of $N_{\wedge A}$ for a directed network. 2. Example of $N_{\wedge A}$ for semi-directed network. 2 continued. Example of $N|_A$ for the previous example.}
% \end{center}
% \end{figure}

\begin{figure}
    \centerline{\includegraphics[width=\textwidth]{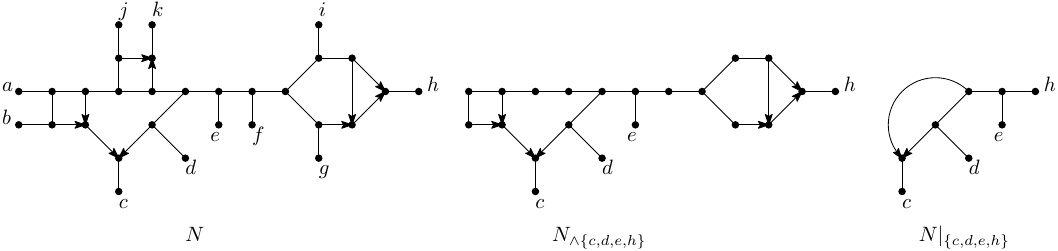}}
    \caption{\label{fig:restr} An example of restricting a semi-directed phylogenetic network~$N$ to a subset of the taxa~$A=\{c,d,e,h\}\subseteq X$. First, all vertices are deleted that are not on a $\wedge$-path between two vertices of~$A$, giving \revnew{the} semi-directed network~$N_{\wedge A}$. Then suppression operations are applied, giving the restriction~$N|_A$, which is a semi-directed phylogenetic network, by Proposition~\ref{prop:restriction_semidirected}. Moreover, since~$|A|=4$, ~$N|_A$ is a quarnet in~$Q(N)$.}
    %and the underlying semi-directed phylogenetic network~$N$. Hence,~$N_d$ is a rooting of~$N$. The top row shows how to obtain $\semi{N_d|_A}$ with~$A=\{c,e,f,i\}$ from~$N_d$ in three steps: the first step deletes all vertices that are not on a $\wedge$-path between vertices in~$A$, the second step applies the suppression operation and the third step replaces all arcs by edges except for arcs entering reticulation vertices. The bottom row shows how to obtain $N|_A$, which is isomorphic to $\semi{N_d|_A}$ \leon{in this case}.}
\end{figure}

\leon{To prove that~$N|_A$ is a semi-directed phylogenetic network, if~$N$ is a semi-directed phylogenetic network, we will use the following lemma, whose proof is deferred to the appendix.}

\begin{lemma}
    \label{lem:wedges} \leon{Consider a network~$N$ on~$X$, leaves~$a,b\in X$ and a reticulation~$v$ with parents~$u,w$. If~$v$ is on a $\wedge$-path in~$N$ between~$a$ and~$b$, then~$u$ is on a $\wedge$-path in~$N$ between~$a$ and~$b$.}
\end{lemma}

\begin{proposition}\label{prop:restriction_semidirected}
Given a semi-directed phylogenetic network~$N$ on~$X$ and a subset~$A\subseteq X$ with~$|A|\geq 2$, the restriction $N|_A$ of~$N$ to~$A$ is a semi-directed phylogenetic network.
\end{proposition}
\begin{proof}
\leon{We first show that~$N_{\wedge A}$ is a semi-directed network. Let~$D$ be a rooting of~$N$. Observe that a non-root vertex~$v$ of~$D$ is on a $\wedge$-path between vertices in~$A$ if and only if the corresponding vertex~$v'$ of~$N$ is on a $\wedge$-path between vertices in~$A$. Hence, $D_{\wedge A}$ contains all vertices of~$N_{\wedge A}$ and possibly one additional vertex; its root. We split the rest of the proof into two cases accordingly}.

\leon{The first case is that~$D_{\wedge A}$ contains the root of~$D$. In this case, $D_{\wedge A}$ contains all vertices of~$N_{\wedge A}$ and exactly one additional vertex; its root~$\rho$. We claim that $\semi{D_{\wedge A}}$ is equal to~$N_{\wedge A}$. To prove this, it remains to show that each edge/arc has the same orientation in $\semi{D_{\wedge A}}$ as in~$N_{\wedge A}$.}

To this end, suppose that~$(u,v)$ is an arc of $\semi{D_{\wedge A}}$. Then~$D$ contains either arc~$(u,v)$ or arcs~$(\rho,u),(\rho,v)$. In either case, since~$v$ is a reticulation in~$D$, $N$ contains an arc~$(u,v)$. Moreover, since~$v$ is in~$N_{\wedge A}$, it follows from Lemma~\ref{lem:wedges} that both incoming arcs of~$v$ in~$N$ are in~$N_{\wedge A}$. Hence, $(u,v)$ is an arc of~$N_{\wedge A}$.

 Now, suppose that~$(u,v)$ is an arc of~$N_{\wedge A}$ and hence of~$N$. Then~$D$ contains either arc~$(u,v)$ or arcs~$(\rho,u),(\rho,v)$. In either case,~$v$ is a reticulation in~$D$. Moreover, since~$v$ is in~$D_{\wedge A}$, it follows from Lemma~\ref{lem:wedges} that both incoming arcs of~$v$ in~$D$ are in~$D_{\wedge A}$. Hence,~$D_{\wedge A}$ contains either arc~$(u,v)$ or arcs~$(\rho,u),(\rho,v)$. In either case, $(u,v)$ is an arc of~$\semi{D_{\wedge A}}$.

\leon{We have now shown that, in the first case, $\semi{D_{\wedge A}}$ is equal to~$N_{\wedge A}$. Hence,~$N_{\wedge A}$ is a semi-directed network.}

\leon{Now consider the second case, i.e., that~$D_{\wedge A}$ does not contain the root of~$D$. In this case, $D_{\wedge A}$ contains exactly the same vertices as~$N_{\wedge A}$. Hence, it follows from Lemma~\ref{lem:wedges} that $(u,v)$ is a reticulation arc of $D_{\wedge A}$ if and only if~$(u,v)$ is a reticulation arc of~$N_{\wedge A}$. This does not imply that $D_{\wedge A}$ is a rooting of $N_{\wedge A}$ because the root may be suppressed when taking the underlying semi-directed network of~$D_{\wedge A}$. Therefore, consider the directed network~$D'$ obtained from $D_{\wedge A}$ by subdividing either of the arcs leaving the root. Then~$\semi{D'}$ is isomorphic to~$N_{\wedge A}$, proving that~$N_{\wedge A}$ is a semi-directed network.}

\leon{We conclude that~$N_{\wedge A}$ is semi-directed in both cases. By Lemma~\ref{lem:semiDirectedPreserved}, it now follows that~$N|_A$ is semi-directed. It is also easy to see that~$N|_A$ is phylogenetic, since otherwise a suppression operation would be applicable.}
\end{proof}

\section{Simple level-2 networks\label{sec:lev2}}
%\subsection{Quarnet encodings}

%{\color{blue} Suggestion: maybe we could swap the order of sections 4 and 5 (so we 
%do the simple level-2 networks first as a sort of warm up) and then we could insert %this definition of quarnets at the start of the new section 4 on simple level-2?}
% LEO: swap sounds good to me, but I like to have the main definitions in the preliminaries

We aim to understand which networks are uniquely determined by their induced set of quarnets. In this section, we shall focus on 
understanding this for some networks that are structurally very simple. To make this more precise, we start \revnew{by} presenting a formal definition of a quarnet.
%\todo{Move 1st sentence to earlier in this section and move rest to section 6, since we only use level-$k$ and generators from sec6 on?}

A \emph{quarnet} is a semi-directed phylogenetic network with exactly four leaves. The set $Q(N)$ of quarnets induced by a semi-directed phylogenetic network~$N$ is defined as
$$
%Q(N) = \left\{ N|_A \,:\, A \in {X \choose 4} \right\}.
Q(N) = \revnew{\left\{ N|_A \,:\, A \subseteq  X, |A|=4 \right\}}.
$$
The leaf set of a quarnet~$q$ is denoted $L(q)$.

\kh{Note that in case $N$ is an %semi-directed level-0 phylogenetic network
\revnew{unrooted phylogenetic tree}
then the quarnets of $N$ are generally called \emph{quartets}.
}

Let~$C$ be a \mj{subclass of the class of semi-directed phylogenetic networks with at least four leaves}. We say that~$C$ is \emph{encoded by quarnets} if for each~$N\in C$ and each semi-directed phylogenetic network~$N'$ \kh{on the same leaf set as $N$ for which} $Q(N)\simeq Q(N')$ holds, \kh{we have} that $N\cong N'$.
We say that~$C$ is \emph{weakly encoded by quarnets} if for all~$N,N'\in C$ \kh{on the same leaf sets and} with~$Q(N) \simeq Q(N')$ holding, we have $N\cong N'$. \kh{Clearly, if $C$ is encoded by quarnets then $C$ is also weakly encoded by quarnets and, as is well known, the class of unrooted phylogenetic trees is encoded by quartets (see e.g. \cite[Theorem 2.7]{dress2012basic}). To help keep terminology at bay, we also say that a member of $C$ is {\em encoded/weakly encoded by \revnew{quarnets}} if $C$ is encoded/weakly encoded by quarnets.
}

 We say that a network~$N$ is \emph{simple} if the mixed graph $N'$ obtained from~$N$ by deleting every leaf %and its incident edges
 is a blob. For a non-negative integer $k$ we call a network $N$ \emph{level-$k$} if each blob of $N$ contains at most~$k$ reticulations, and we call $N$ \emph{strict level-$k$} if, in addition, it contains a blob with exactly~$k$ reticulations. Note that a semi-directed level-0 phylogenetic network is an unrooted phylogenetic tree in the usual sense (see e.\,g.\,\cite{semple2003phylogenetics} for
 more details concerning such trees) and that, by definition, a simple network is strict level-$k$, for some~$k\geq 1$.
\kh{For example, the directed phylogenetic network~$N_d$ in Figure~\ref{fig:generator} is simple and so is the semi-directed phylogenetic network $N$ in the same figure. Furthermore, both networks are strict level-2.}

\begin{figure}[h]
\begin{center}
\includegraphics[]{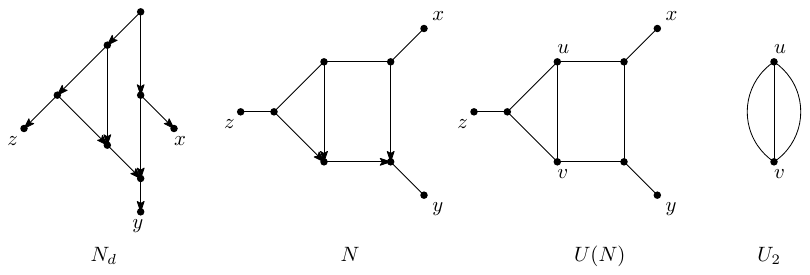}
\caption{\kh{A directed, simple, strict level-2 phylogenetic network on $X=\{x,y,z\}$, a semi-directed simple, strict level-2 phylogenetic network $N$ on $X$, \revnew{the underlying graph~$U(N)$ of~$N$ and the undirected level-$2$ generator~$U_2$}.}
\label{fig:generator}
}
\end{center}
\end{figure}

%\kh{For $k>1$, we call a \leon{mixed graph}
% it is not necessarily a semi-directed network because it will have indegree-2 outdegree-0 vertices
%that can be obtained from a semi-directed, simple, phylogenetic strict level-$k$ network by deleting all leaves and applying vertex suppression operations~(V1) and (V2) a {\em (semi-directed) level-$k$ generator}.} The {\em sides} of a semi-directed level-$k$ generator \kh{are defined as} the arcs, edges and the degree-2 vertices in the generator \leon{(which have indegree~$2$ outdegree~$0$). See Figure~\ref{fig:generator}(iii) for an example}.
%
%\begin{figure}[h]\label{fig:generator}
%\begin{center}
%\includegraphics[width = 0.5\textwidth]{generator}
%\caption{\kh{(i) A directed, simple, strict level-2 network on $X=\{x,y,z\}$. (ii) A semi-directed simple level-2 network $N$ on $X$. {\bf Add generator!}}}
%\end{center}
%\end{figure}

%\subsection{Quarnet encodings}

% \mj{For two sets of quarnets $Q$ and $Q'$, we write $Q \simeq Q'$ to denote the existence of a bijective function $\phi:Q \rightarrow Q'$ such that for all $q \in Q, \phi(q) \cong q$.}

% We assume that $|X|\ge 4$. Let $N$ be a semi-directed \leon{phylogenetic} network on $X$.
% Recall that $Q(N)$ denotes the set of quarnets of $N$, i.e., the set of 
% 4-leaved semi-directed networks\todo{seems unnecessary to repeat the definition here}
% $$
% Q(N) = \left\{ N|_A \,:\, A \in {X \choose 4} \right\}.
% $$

\kh{To be able to prove Lemma~\ref{lem:level-1}, we require further concepts. Suppose} $T$ is an unrooted phylogenetic tree on $X$ with $|X|\geq 4$. 
A {\em cherry} in $T$ is a pair of leaves of $T$
that are adjacent to the same vertex of $T$. 
% If $T$ has
% four leaves, we call it a {\em quartet}, and 
If $T$ contains precisely
two cherries, we call it a {\em caterpillar tree}.

% We note
% that the set of \leon{quarnets (in this case also called quartets)} $Q(T)=\{T|_A \,:\, A \in {X \choose 4}\}$ 
% encodes $T$
%, that is, if $T'$ is an unrooted, binary phylogenetic tree on $X$ with $Q(T')=Q(T)$, then $T$ is isomorphic to $T'$

%We say that a class $\mathcal C$ of semi-directed networks is {\em weakly encoded by quarnets} if for all $N,N' \in \mathcal C$, if $Q(N)=Q(N')$, then $N$ is isomorphic to $N'$.

%\kh{For $k>1$, we call a \leon{mixed graph}
% it is not necessarily a semi-directed network because it will have indegree-2 outdegree-0 vertices
%that can be obtained from a semi-directed, simple, phylogenetic strict level-$k$ network by deleting all leaves and applying vertex suppression operations~(V1) and (V2) a {\em (semi-directed) level-$k$ generator}.} The {\em sides} of a semi-directed level-$k$ generator \kh{are defined as} the arcs, edges and the degree-2 vertices in the generator \leon{(which have indegree~$2$ outdegree~$0$). See Figure~\ref{fig:generator}(iii) for an example}.

\begin{lemma}\label{lem:level-1}
    The class of semi-directed, simple, strict level-1 \leon{phylogenetic} networks \kh{with at least four leaves}
    is weakly encoded by quarnets.
\end{lemma}
\begin{proof}
Suppose that $N$ \kh{is a semi-directed, simple, strict level-1 phylogenetic network with at least four leaves. Let $N'$ be a semi-directed, simple, strict level-1 phylogenetic network}
on the leaf set $X$ of $N$ with $Q(N)\simeq Q(N')$. We \kh{need} to
show that $N'$ is isomorphic to $N$. 
If $|X|=4$, this is trivial, so suppose $|X|\ge 5$.

\kh{We start with a central observation. Suppose $M$ is} a semi-directed, simple, strict level-1 network and $x$ is a leaf of $M$ that is adjacent to \kh{the unique reticulation $r$ of $M$. Then, by the 
definition of a quarnet induced by $M$,} every quarnet in $Q(M)$ is either a semi-directed, simple,  strict level-1 network such that $x$ is also adjacent to \kh{$r$}, or it is a \kh{phylogenetic tree} whose leaf set 
does not contain $x$.
In view of this observation, if $x \in X$ is the leaf in $N$ that is adjacent to
the \mj{unique} reticulation in $N$, then since $Q(N)\simeq Q(N')$ it follows that 
$x$ is adjacent to the  \mj{unique}  reticulation in $N'$. 

Now, let $P=X - \{x\}$. \kh{For every leaf $y\in P$ let $v_y$ denote the vertex in $N$ adjacent with $y$.} Suppose
that $a,b,c,d \in P$ 
are such that when \kh{traversing} the cycle in $N$ 
%clockwise
we have the path $v_a,v_b,r=v_x,v_c,v_d$.
Consider the set 
$$
  Q=  \{ N|_A \,:\, A \in {P \choose 4} \}.
$$
By the above observation, it is straight-forward to see that the caterpillar tree $C$ on $P$ with 
cherries $\{a,b\}$ and $\{c,d\}$ \kh{is encoded by $Q$. Since $Q(N) \simeq Q(N')$, it follows that} $N'$
must induce a caterpillar tree on \kh{$P$ that is isomorphic with $C$}.
By considering the two quarnets in $Q(N)$
on the sets $\{a,b,x,c\}$ and $\{b,x,c,d\}$,
it follows that the \kh{order of the leaves $a,b,x,c,d$ 
 in $N$ induced by the path $v_a,v_b,r,v_c,v_d$ must be the same as in $N'$}. Hence, $N'$ is isomorphic to $N$.
\end{proof}

\kh{To be able to study weak encodings of  level-2 networks, 
% I moved the definition of a generator to the preliminaries
we refer to the graph obtained from a phylogenetic network~$N$ by removing all directions as the {\em underlying graph} of $N$ and denote it by $U(N)$, \revnew{see Figure~\ref{fig:generator}}. Note that $U(N)$ is indeed a graph \revnew{(and not a multi-graph)} because $N$ is a phylogenetic network and so cannot contain parallel arcs.
Note that so-called undirected phylogenetic networks
are precisely the undirected graphs $G$ for which there exists a semi-directed network $N$ such that $G$ and $U(N)$ are isomorphic and the leaf sets of $G$ and $N$ coincide. Calling a multi-graph with two vertices and three parallel edges joining these vertices an {\em undirected level-2 generator} and
canonically extending the notion of a simple, strict level-2 network to undirected phylogenetic networks then, by \cite[Fig. 4]{van2018leaf}, every undirected, simple, strict level-2 phylogenetic network on $X$ can be obtained from an undirected level-2 generator
by subdividing the %arcs
\revnew{edges}
of the generator to obtain three paths $P_1,P_2,P_3$ with end vertices $u$ and $v$ that intersect pairwise only at $u$ and $v$, such that (i) at least two of these paths have length at least 2, and (ii) for $i=1,2,3$, every vertex $w \in V(P_i)\setminus\{u,v\}$ is adjacent to a leaf in $X$. }

\kh{Motivated by the above, we call for all $k\geq 2$ the \leon{mixed graph}
% it is not necessarily a semi-directed network because it will have indegree-2 outdegree-0 vertices
that can be obtained from a semi-directed, simple, strict level-$k$ phylogenetic network $N$ by deleting all leaves and applying vertex suppression operations~(V1) and (V2) a {\em (semi-directed) level-$k$ generator for $N$} and denote it by $gen(N)$. More generally, we call a mixed graph $G$ a {\em level-$k$ generator} if there exists a semi-directed, simple, strict level-$k$, phylogenetic network $N$ such that $G$ and $gen(N)$ are isomorphic. See Figure~\ref{fig:level2gens-labelled} for two semi-directed level-2 generators. To see that these are in fact all semi-directed level-2 generators 
(Lemma~\ref{lem:level-2-generators}) we use that every semi-directed level-2 generator can be obtained from an undirected, simple, strict level-2 phylogenetic network.
} 
%
%\begin{figure}[h]
%\begin{center}
%\includegraphics[width = 0.5\textwidth]{generator}
%\caption{\label{fig:generator}
%\kh{(i) A directed, simple, strict, level-2 phylogenetic network on $X=\{x,y,z\}$. (ii) A semi-directed simple level-2 phylogenetic network $N$ on $X$. 5{\bf Add level-2 generators!}}}
%\end{center}
%\end{figure}
%
 %
\begin{figure}[h]
    \centering
    \includegraphics{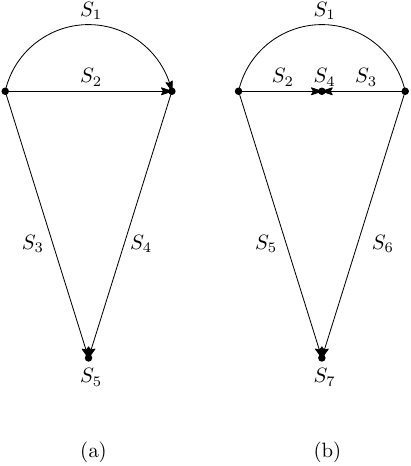}
    \caption{The semi-directed level-2 generators with sides labelled.} %{\bf We need to add "(a)" and "(b)" below the generators so that we can properly refer to them in the proof of Lemma~\ref{prop:sides}}.
    \label{fig:level2gens-labelled}
\end{figure}

\begin{lemma} \label{lem:level-2-generators}
    The semi-directed level-2 generators are as pictured in Figure~\ref{fig:level2gens-labelled}.
\end{lemma}
\begin{proof}
Suppose that $G$ is one of the mixed graphs in Figure~\ref{fig:level2gens-labelled}. We need to show that there exists a semi-directed, simple, strict level-2, phylogenetic network $N$ such that $gen(N)$ and $G$ are isomorphic. We can obtain $N$ as follows. In case (a), subdivide~$S_2$ into an edge, a new vertex~$v$, and an arc and then add a leaf adjacent to~$v$. In either case, add a leaf adjacent to each outdegree-0 reticulation. To see that~$N$ is semi-directed, note that you can obtain a directed network by subdividing~$S_1$ by the root and directing all edges away from the root. Hence,~$N$ is a semi-directed, simple, strict level-2, phylogenetic network $N$ such that $gen(N)$ and $G$ are isomorphic. It follows that the mixed graphs in Figure~\ref{fig:level2gens-labelled} are semi-directed level-2 generators.

To show that these are all semi-directed level-2 generators, consider a semi-directed, simple, strict level-2 phylogenetic network~$N$. Observe that $U(N)$ is an undirected, simple, strict level-2 phylogenetic network. Let $u$ and $v$ denote the vertices of the undirected level-2 generator. Let $P_1,P_2,P_3$ denote the three paths in $U(N)$ from $u$ to $v$.

Observe that~$N$ has, by definition, precisely two reticulations. Call these reticulations~$p$ and~$q$.
%(because any directed, simple, strict level-2 phylogenetic network has two reticulations -- see e.g. \cite[Section 6]{vIM2014}),
%we select two non-leaf vertices $p$ and $q$ 
%    in $U(N)$ and orient, for each of $p$ and $q$, the two non-pendant edges in $U(N)$ that contain $p$ and $q$ to make them the two reticulations of $N$.
If $\{p,q\}\cap \{u,v\}=\emptyset $, then there must exist distinct $i,j\in \{1,2,3\}$ such that
$p$ is a vertex on $P_i$ and $q$ is a vertex on $P_j$ as otherwise it would not be possible to orient the edges in $U(N)$ so as to obtain a semi-directed, simple, strict level-2 phylogenetic network with reticulations~$p$ and~$q$. Similarly, it is not possible that $\{p,q\}=\{u,v\}$.
 Hence, we must either have that $\{p,q\}=\{u,w\}$ or $\{p,q\}=\{v,w\}$, with 
    $w \notin\{u,v\}$ a vertex on $P_i$ some $1 \le i \le 3$, 
    or that $\{p,q\} = \{w,w'\}$ with $\{w,w'\}\cap \{u,v\}=\emptyset$ and $w$ a vertex on $P_i$ and $w'$ a vertex on $P_j$, where $i ,j\in\{1,2,3\}$. 
In the first case, $gen(N)$ is the mixed graph in Figure~\ref{fig:level2gens-labelled}(a). In the second case, $gen(N)$ is the mixed graph in Figure~\ref{fig:level2gens-labelled}(b).
\end{proof}

\kh{As we shall see, the next result (Proposition~\ref{prop:sides})
is central for showing that the class of semi-directed simple, strict level-2 phylogenetic networks with at least four leaves is weakly encoded by quarnets (Theorem~\ref{thm:lev2}). To be able to state and prove  it, we again require further definitions.}

\kh{\revnew{The following definitions are illustrated in Figure~\ref{fig:genex}.} Suppose that $N$ is a semi-directed, simple, strict level-2 phylogenetic network. Then we call} the arcs, edges and the degree-2 vertices in $gen(N)$ \leon{(which have indegree~$2$ outdegree~$0$) the {\em sides} of $gen(N)$.} 
\revnew{For example, the sides of~$gen(N)$ in %Figure~\ref{fig:genex}
the example are the arcs~$a_1,a_2,a_3,a_4$ and the vertex~$v_6$.} \kh{If a side $S$ of $gen(N)$ is an arc/edge, then we denote by $P(S)$ the %directed/undirected 
\revnew{semi-directed}
path in $N$ such that when deleting all leaves of $N$ adjacent with a vertex of $P(S)$ and suppressing all resulting vertices of $P(S)$ with overall degree two, we obtain $S$. \revnew{In the example, we have $P(a_3)=(v_1,v_2,v_4,v_6)$, $P(a_1)=(v_1,v_3,v_5)$ and $P(a_2)=(v_1,v_5)$ (where~$a_1,a_2$ could be swapped).}
Note that $P(S)$ could be an arc/edge in $N$ \revnew{(such as $P(a_4)=(v_5,v_6)$ 
%in Figure~\ref{fig:genex})
}.
In case $S$ is a vertex, then we also refer to $S$ as $P(S)$ \revnew{(e.g. $P(v_6)=v_6$
%in Figure~\ref{fig:genex})
}.
We say that a leaf $x$ of $N$ is {\em hanging off $S$ in $N$} if either $S$ is a vertex of $gen(N)$ with overall degree two and $N$ contains the edge $\{S,x\}$ or $S$ is an arc/edge in $gen(N)$ and there exists a vertex $v$ on $P(S)$ such that $\{v,x\}$ is an edge of $N$. \revnew{In the example,~$z$ is hanging off~$a_1$ and~$w$ is hanging off~$v_6$.} We denote the set of leaves of $N$ hanging off $S$ by $P_S$. \revnew{In the example, $P_{a_3}=\{x,y\}$ and~$P_{a_4}=\emptyset$.}
%Note that $P_S=\emptyset$ if $P(S)$ is an arc/edge. 
Finally,} 
we say that two semi-directed,  simple, strict level-2 phylogenetic networks $N$, $N'$ 
are {\em isomorphic up to sides} if there is some \kh{(mixed graph)} isomorphism $\phi$ 
between \kh{$gen(N)$ and $gen(N')$} so that for any side $S$ in \kh{$gen(N)$}, the \kh{leaf sets $P_S$ and $P_{\phi(S)}$
are equal. \revnew{In the example,~$N$ and~$N'$ are isomorphic up to sides.}}

\begin{figure}[h]
    \centering
    \includegraphics{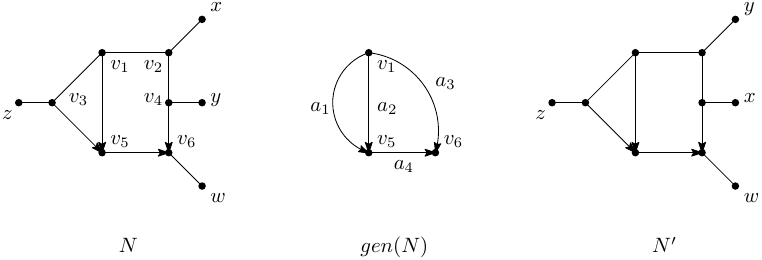}
    \caption{\revnew{Two semi-directed, simple, strict level-$2$ phylogenetic networks~$N$ and~$N'$ that are isomorphic up to sides, together with their level-$2$ generator $gen(N)\cong gen(N')$.}
    \label{fig:genex}}
\end{figure}

\begin{proposition}\label{prop:sides}
	Suppose that $N$ \revnew{and}~$N'$ are semi-directed, simple, strict level-2 \leon{phylogenetic} networks with at least four leaves 
	that are isomorphic up to sides. If \leon{$Q(N)\simeq Q(N')$, then $N\cong N'$}.
\end{proposition}
\begin{proof}
\kh{Suppose that $Q(N)\simeq Q(N')$ 
	%that $N,N'$ have the same generator as in Figure~\ref{fig:level2gens-labelled}(a) or (b), %that $S$ is a side of $gen(N)$ 
	and 
	that $\phi$ is an isomorphism from $gen(N)$ to $gen(N')$. By Lemma~\ref{lem:level-2-generators} 
 it follows that $gen(N)$ and $gen(N')$ are either both as depicted in Figure~\ref{fig:level2gens-labelled}(a) or they are both as depicted in Figure~\ref{fig:level2gens-labelled}(b).
	%such that $P_S=P_{\phi(S)}$. 
}

\kh{	{\bf Claim:} Suppose $S$ is a side of $gen(N)$ for which $P_S\not=\emptyset$. If $S$ is an arc then, irrespective of Case~(a)
	or (b) holding for $N$ in Figure~\ref{fig:level2gens-labelled}, the order in which the elements 
	in $P_S$ hang off $\phi(S)$  in $N'$ relative to the direction of $\phi(S)$
	is the same as the order in which they hang off $S$ in $N$ relative to the direction of  $S$. If $S$ is the unique edge
	in Figure~\ref{fig:level2gens-labelled}(b), then the order in which the elements 
	in $P_S$ hang off $\phi(S)$ in $N'$ 
	is the same as the order in which they hang off $S$ in $N$, up to reversing the whole ordering.
}

\kh{	{\em Proof of Claim:} 
	Suppose $S$ is an arc in $gen(N)$ such that $P:=P_S\not=\emptyset$. If $|P|=1$, then the  claim trivially holds. So assume that $|P|\geq 2$. We distinguish between the cases that $|P|=2$, that $|P|=3$, and that $|P|\geq 4$. 
}
 
\kh{	If $|P|=2$, then we consider a 4-subset $A$ containing $P$
	which is defined as follows. If Figure~\ref{fig:level2gens-labelled}(a) holds, then $N$ has a unique reticulation. Let $x \in X$ be the 
	leaf below that reticulation and let $y \in X-(\leo{P} \cup \{x\})$. If Figure~\ref{fig:level2gens-labelled}(b) then $N$ has two reticulation\revnew{s}. Let $x,y \in X$ be
	the leaves below the two reticulations, respectively. In either case, 
	let $A=P \cup \{x,y\}$. Then  $N|_A\in Q(N)\simeq Q(N')$. That the claim holds is
	straight-forward to see. 
}

\kh{
If $|P|=3$, then we consider two 4-subsets $A,B$ of $X$
	which are defined as follows. Suppose first that $gen(N)$ is as in Figure~\ref{fig:level2gens-labelled}(a) and that $x $ is the 
	leaf of $N$ below the unique reticulation $r$ of $N$. Then the size of $A:=P\cup\{x\}$ is four since  $|P|=3$. Moreover, 
	if $P$ equals $P_{S_3}$ or $P_{S_4}$, then we choose a leaf $a $  in $P_{S_1}$ or $P_{S_2}$ which must exist as $N$ is strict level-2. To obtain $B$,  we choose leaves $b,c$
	in $P$ such that the unique vertex in $N$ adjacent with $b$ is adjacent with $r$
	as well as with the unique vertex in $N$ adjacent with $c$. Finally, we  put $B=\{a,b,c,x\}$. 
}

\kh{	Assume for the remainder of this case that $gen(N)$ is as in Figure~\ref{fig:level2gens-labelled}(b). Let $x,y $ be the leaves of $N$ such that $x$ is below one reticulation of $N$ and $y$ is below the other. Then we put $A=P \cup \{x\}$ and $B= P \cup \{y\}$ which both clearly have size four since $|P|=3$.
}

\kh{	In either of the above two cases, $N|_A, N|_B \in Q(N)\simeq Q(N')$ follows. That the claim holds is a straight-forward consequence. 
}
 
\kh{	If $|P|\ge 4$, then consider the set $R$ of 
	{\em quartets} obtained by restricting $N$ to all possible 4-subsets of $P$. 
	Then $R$ must be the set of quartets induced by some caterpillar tree $T$ with
	leaf set $P$. Since $Q(N)\simeq Q(N')$ it follows that the leaves in $P$ are hanging off $S$ in $N$ in the same ordering as the leaves of $P$ are hanging off $\phi(S)$ in $N'$, up to 
	reversal of the two leaves in each of the cherries in $T$ \leo{and up to reversing the whole ordering}. 
	The claim now follows by considering, in addition 
	to $R$, the set of all quarnets with leaf set $\{a,b,c,x\}$,
	where  $a$ and $b$ form a cherry in $T$, $c \in P- \{a,b\}$ and $x$ is a
	leaf below a reticulation in $N$. This completes the proof of the Claim.
}
 
\kh{	If $gen(N)$ is as in Figure~\ref{fig:level2gens-labelled}(a), then  the lemma follows by applying  the Claim to each side $S$ of $gen (N)$ for which $P_S\not=\emptyset$ holds.
	If $gen(N)$ is as in Figure~\ref{fig:level2gens-labelled}(b), then the lemma follows again by
	applying the Claim to each side $S$  of $gen(N)$ for which $P_S\not=\emptyset$ holds  
	in case  $P_{S_1}=\emptyset$, that is, no leaf of $N$ is hanging off $S_1$ in $N$. Furthermore,
	the lemma follows by applying the Claim to side $S_1$ of $gen(N)$
	if $P_{S_i}=\emptyset$ holds for all $i\in\{2,3,5, 6\}$, that is, other than the leaves of $N$ hanging off the two reticulations of $N$, every leaf of $N$ is hanging off $S_1$ in $N$. 
 }
	
\kh{Assume for the remainder that $P_{S_1}\not=\emptyset $ and that there exists some $i\in\{2,3,5,6\}$ such that $P_{S_i}\not=\emptyset$. To see that the order in which the elements 
	in $P_{S_1}$ are hanging off $\phi(S_1)$ in $N'$ 
	is the same as the order in which they are hanging off $S_1$ in $N$, we may assume without loss of generality that $i=2$. Choose leaves  $a\in P_{S_1}$ and 
	$b\in P_{S_2}$ such that there exists a vertex $w$ in $N$ such that the shortest path from $a$ to $b$ in $N$ contains $w$.  Since $N$ is a semi-directed, simple, strict level-2 network  there must exist a leaf $x$ of $N$ that is adjacent with one reticulation of $N$ and a leaf $y$ of $N$ that is adjacent with the other. Then $N|_{\{a,b,x,y\}}$
	is a quarnet in $Q(N)\simeq Q(N')$. Thus, the shortest path from $a$ to $b$ in $N'$ contains $\phi(w)$. Since, by the Claim,  the order in which the elements 
	in $P_{S_1}$ are hanging off $\phi(S_1)$ in $N'$ 
	is the same as the order in which they are hanging off $S_1$ in $N$, up to reversing the whole ordering, it follows that $N\cong N'$. This completes the proof of the lemma.
 }
\end{proof}

\begin{lemma}\label{prop:level-2}
The class of semi-directed, simple, strict level-2 \leon{phylogenetic} networks \leo{with at least four leaves} is weakly encoded by quarnets.
\end{lemma}
\begin{proof}
\kh{Suppose that $N$ is a semi-directed, simple, strict level-2 phylogenetic network with at least four leaves.
Let $X$ be the leaf set of $N$ and let $N'$ be a semi-directed, simple, strict level-2 phylogenetic network on $X$ such that}
    $Q(N) \simeq Q(N')$. We \kh{need} to show that $N$ and $N'$ \kh{are} isomorphic.
    By Lemma~\ref{prop:sides}, it suffices to show that 
    $N$ and $N'$ are isomorphic up to sides. 
    
    First note that $N$ and $N'$ must have isomorphic generators.
    Indeed, there must be some 4-subset $A$ of $X$ so that $N|_A$ (and thus $N'|_A$) 
    \kh{is a semi-directed, simple, strict level-2 phylogenetic network. Since $Q(N)\simeq Q(N')$
    it follows that $gen(N|_A)$ and $gen(N'|_A)$ are isomorphic. By Lemma~\ref{lem:level-2-generators}, it follows that $gen(N)$ and $gen(N')$ must be isomorphic, as required}. 
    
    We next show that there
    exists some isomorphism $\phi$ from \kh{$gen(N)$ to $gen(N')$  so that if $S$ is any side in $gen(N)$ with $P_S\not=\emptyset$ then $P_S=P_{\phi(S)}$.} 
    %and $P$ is the set of leaves hanging off $S$, 
    %then the set of leaves hanging off $\phi(S)$ %is $P$. 
    To show that such an isomorphism $\phi$ exists, we \kh{distinguish between the cases that $gen(N)$ is as in Figure~\ref{fig:level2gens-labelled}(a) and that $gen(N)$ is as in Figure~\ref{fig:level2gens-labelled}(b). Put $P_i=P_{S_i}$, for all $i$}.

Case~(a): Note that in this case, there are exactly two isomorphisms from 
\kh{$gen(N)$ to $gen(N')$}: the identity and one
that swaps the sides $S_1$ and $S_2$ of $gen(N)$.
Now, fix $x\in P_5$ and an arbitrary leaf~$y\in P_i$ with~$i\in\{1,2\}$. Then considering any quarnet containing leaves~$x$ and~$y$ we see that~$y$ hangs off $\phi(S_1)$ or off $\phi(S_2)$ in~$N'$, for any isomorphism~$\phi$. Choose $\phi$ such that~$y$ hangs off $\phi(S_i)$. Now consider any leaf $z\in P_i\setminus\{y\}$, $i\in\{1,2,3,4\}$. Then considering any quarnet whose leaf set contains leaves~$x,y$ and~$z$, we see that~$z$ hangs off $\phi(S_i)$ in $N'$. This completes the proof in this case.

Case~(b): Note that in this case there are exactly four isomorphisms from \kh{$gen(N)$ to $gen(N)'$}: the identity, 
$(S_2S_3)(S_5S_6)$, $(S_2S_5)(S_4S_7)(S_3S_6)$ and $(S_2S_6)(S_4S_7)(S_3S_5)$ (given
as a combination of swaps, where $(S_i,S_j)$ denotes swapping sides $S_i$ and $S_j$).
Now, let $x\in P_4$ and $y\in P_7$. First suppose $P_1=X\setminus \{x,y\}$. For any leaf~$z\in P_1$ it follows, by considering an arbitrary quarnet whose leaf set contains $x$, $y$ and~$z$ that~$z$ hangs off~$\phi(S_1)=S_1$ in~$N'$. So the lemma holds in this case. 

Assume for the remainder that 
there exists $q\in X\setminus (P_1\cup\{x,y\})$. 
Then~$q\in P_i$ with~$i\in\{2,3,5,6\}$. Hence, by 
considering any quarnet whose leaf set contains $x$, $y$ and~$q$, we 
see that~$q$ hangs off one of $\phi(S_2)$, $\phi(S_3)$, $\phi(S_5)$, $\phi(S_6)$ in $N'$. 
Choose the isomorphism $\phi$ such that~$q$ hangs off $\phi(S_i)$ in $N'$. 
Then, for any leaf~$z\in P_i$, $i\in\{1,2,3,5,6\}$, it 
follows, by considering the quarnet whose leaf set contains $x$, $y$, $z$, $q$, 
that~$z$ hangs off~$\phi(S_i)$ in~$N'$. 
This completes the proof of the lemma in this case too.
\end{proof}

\revnew{The next theorem corresponds to Lemma~\ref{prop:level-2} without the ``strict'' restriction.}

\begin{theorem}\label{thm:lev2}
The class of semi-directed, simple,  
level-2, \leon{binary phylogenetic} networks with at least four leaves is weakly encoded by quarnets. %Moreover, the class of simple semi-directed level-$3$ \leon{binary, phylogenetic} networks with at least four leaves is \emph{not} weakly encoded by quarnets.
% Leo: moved the last statement to the intro
\end{theorem}
\begin{proof}
Suppose that $N$ and $N'$ are semi-directed, simple, level-2 \leon{phylogenetic} networks 
    on $X$ with $Q(N)\simeq Q(N')$. We want to
    show that $N$ is isomorphic to $N$. 
    
First note that, if $N$ is strict level-2 and $N'$ is strict level-1, then 
we can clearly pick some
%$A \in {X \choose 4}$
\revnew{$A\subseteq X,|A|=4$,}
so that the quarnet 
$N|_A$ \kh{is a strict level-2 network}, which is impossible since $N'|_A$ must be a level-1 \kh{network}.
By symmetry, it follows that 
both $N$ and $N'$ must be a strict level-1 or a strict level-2 \kh{network}. 
The theorem now follows immediately by 
applying Lemmas~\ref{lem:level-1} and \ref{prop:level-2}.
\end{proof}

\section{Blob trees\label{sec:cuts}}

%{\color{blue} Instead of starting with definitions, maybe we could 
%start by saying that in this section we show that 
%the blob tree of a network N can be found from Q(N), and then explicitly state this 
%more formally as corollary of Thm 4.1?}

\begin{figure}
    \centerline{\includegraphics{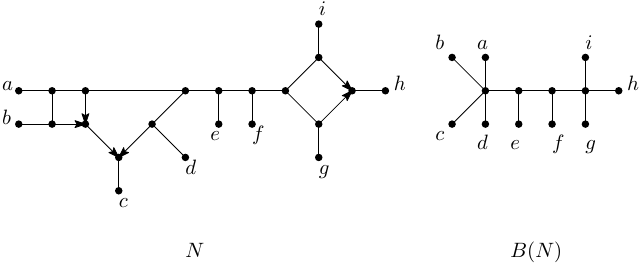}}
    \caption{\label{fig:intro} A semi-directed phylogenetic network~$N$ and the blob tree~$B(N)$ of~$N$.}
\end{figure}

\rev{Observe that a directed (respectively semi-directed) network is phylogenetic precisely if it has no parallel arcs and contracting each blob into a single vertex gives a directed (respectively undirected) phylogenetic tree. The tree obtained in this way is called the \emph{blob tree}~$B(N)$ of a network~$N$, see Figure~\ref{fig:intro}.} In this section, we show that the blob tree of a semi-directed phylogenetic network is uniquely determined by the quarnets of the network. This will be a direct consequence of Theorem~\ref{thm:splitsCondition}, which characterizes the splits of a semi-directed phylogenetic network using its quarnets. Note that this theorem does not put any restriction on the level.

A \emph{cut-edge} of a semi-directed network is an edge whose removal disconnects the network. \kh{We call a bipartition $\{A,B\}$ of $X$ into two non-empty subsets $A$ and $B$ a {\em split} of $X$ and
denote it by $A|B$ where the order of $A$ and $B$ does not matter. We call a split $A|B$ trivial if $|A|=1$ or $|B|=1$.}

Given a semi-directed network $N$ on $X$ and a 
\kh{split $A|B$ of $X$}
%\vm{bipartition} $\{A,B\}=A|B =B|A$ of $X$ into non-empty subsets $A,B$ of $X$,
we say that $A|B$ is a \emph{cut-edge split (CE-split)} in $N$ if there exists a cut-edge $\{u,v\}$ of $N$ such that its removal gives two connected mixed graphs with leaf-sets $A$ and $B$. We say a CE-split $A|B$ is \emph{trivial} if $|A|=1$ or $|B|=1$. Observe that a semi-directed \leon{phylogenetic} network is simple if and only if it has no nontrivial CE-splits.

We will show in this section that we can detect splits in a semi-directed \leon{phylogenetic} network by looking at its quarnets, using the following theorem:

\begin{theorem}\label{thm:splitsCondition}
   Let  $N$  be a semi-directed, binary
    %level-$2$ 
    phylogenetic network on $X$ and $A|B$ a \kh{split} of $X$. Then $A|B$ is a CE-split in $N$ if and only if one of the following holds:
    \begin{itemize}
        \item \mj{$A|B$ is a trivial split} \kh{of $X$}; or
        \item \mj{$A|B$ is non-trivial} and for any \vm{pairwise distinct elements} $a_1,a_2 \in A, b_1,b_2 \in B$, $\{a_1,a_2\}|\{b_1,b_2\}$ is a CE-split in $N|_{\{a_1,a_2,b_1,b_2\}}$.
    \end{itemize} 
\end{theorem}

The main challenge in proving Theorem~\ref{thm:splitsCondition} will be to show that when \mj{$A|B$ is non-trivial} and $N$ is simple (and therefore $A|B$ is not a CE-split in $N$), there exist $a_1,a_2 \in A, b_1,b_2 \in B$ for which $\{a_1,a_2\}|\{b_1,b_2\}$ is \emph{not} a CE-split in $N|_{\{a_1,a_2,b_1,b_2\}}$.
To show this, we first prove some results concerning directed networks:

\begin{lemma}\label{lem:crossyArc}
    Let $N$ be a simple directed  \leon{phylogenetic} 
    network on $X$ \mj{with at least one reticulation}.
    \leo{If~$v$ is a vertex of~$N$ that is not the root, not a leaf and not a leaf-reticulation, then}
    % Then for any non-leaf, non-root vertex $v$, either $v$ is a leaf-reticulation, or 
    there exists an arc $(u',v')$ with $v \notin \{u',v'\}$ such that $v'$ is  below %$u$
    \leo{$v$} and $u'$ is not below $v$. \mj{In particular, $v'$ is a reticulation.} 
\end{lemma}
\begin{proof}
    %Assume $v$ is not a reticulation-leaf.
    Let  $Y$ denote the set of non-leaf vertices in $N$ that are not below $v$, and let $Z$ denote the set of non-leaf vertices in $N$ \leo{strictly} below $v$.
    Since $v$ is not the root, $Y$ is nonempty. In addition, since $v$ is \leo{not a leaf and} not a leaf-reticulation, $Z$ is nonempty. Then since $N$ is simple, the underlying undirected graph of $N$ has a path \mj{starting at a vertex $Y$ and ending at a vertex in $Z$} 
    that does not include $v$. 
    It follows that there exist adjacent vertices $u'$ in $Y$, $v'\in Z$.
    Since $v'$ is below $v$ and $u'$ is not, $N$ does not contain the arc $(v',u')$. So $N$ must contain the arc $(u',v')$, as required.
\end{proof}

\begin{lemma}\label{lem:wavyPaths}
    Let $N$ be a simple directed \revnew{strict} level-$k$ \leon{phylogenetic} network on $X$ for $k \geq 1$. Then for any arc $(u,v)$ in $N$ \leo{with~$v$ not a leaf}, there exist vertices $u^*, r$ and directed paths $P, Q$ in $N$ such that:
    \begin{itemize}
        \item $P$ and $Q$ are arc-disjoint paths from $u^*$ to $r$;
        \item $P$ contains the arc $(u,v)$; and
        \item $r$ is a leaf-reticulation.
    \end{itemize}
\end{lemma}
\begin{proof}
    Suppose first that $v$ is a leaf-reticulation, and let $(u',v)$ be the other incoming arc of $v$. Then let $r = v$ and let $u^*$ be any lowest common ancestor of $u$ and $u'$. Let $P$ be a directed path consisting of a directed path from $u^*$ to $u$ extended with the arc $(u,v)$, and let $Q$ be a directed path consisting of a directed path from $u^*$ to $u'$ extended with the arc $(u',\leo{v})$. Then $P$ and $Q$ are arc-disjoint paths from $u^*$ to $r=v$ (any overlap would imply that $u$ and $u'$ have a common ancestor strictly below $u^*$) and $P$ contains $(u,v)$.

    Now assume that $v$ is
    %not a leaf-reticulation. 
    \mj{either a tree node or reticulation that is not a leaf-reticulation.}
    We generate a sequence of vertices $v_1,u_1,\dots ,v_{s-1},u_{s-1}, v_s$, as follows.
    Initially set $v_1: =v$ and $i = 1$.
    While $v_i$ is not a leaf-reticulation, by Lemma~\ref{lem:crossyArc} there exists at least one arc $(u',v')$ with \mj{reticulation}
    $v'$ \leo{strictly} below $v_i$ and $u'$ %$ \neq v_i$
    not below $v_i$.
    Choose such an arc $(u',v')$ with lowest $v'$, and let $u_i:= u'$, $v_{i+1} = v'$. 
    Observe that any directed path from $v_i$ to $v_{i+1}$ is arc-disjoint from any directed path ending with $(u_i,v_{i+1})$.
    Now increase $i$ by $1$ and repeat.
    If $v_i$ is a leaf-reticulation, then set $s: = i$ and terminate. 

    Since $v_{i+1}$ is strictly below $v_i$ for each $i$, this process must terminate \mj{because $N$ only has finitely many vertices. Let $x$ be the leaf adjacent to $v_s$ --- see Figure~\ref{fig:wagglyPath}.}

    Note that $v_j$ is below $v_i$ for all $1\leq i < j \leq s$.
    Note also that for each $i \in \{1,\dots, s-2\}$, the vertex $u_{i+1}$ is below $v_i$. Indeed, if this is not the case then $(u_{i+1}, v_{i+2})$ is an arc with $v_{i+2}$ below $v_i$ and $u_{i+1}$ not below $v_i$, which contradicts our choice of $v_{i+1}$ as a lowest vertex for which such an arc exists.
    So there exists a path from $v_i$ to $u_{i+1}$ for each $i \leq s-2$, and an arc from $u_i$ to $v_{i+1}$ for each $i \leq s-1$.

    Now let $u^*$ be a lowest common ancestor of $u$ and $u_1$.

    We can now form $P$ by combining the following directed paths --- see Figure~\ref{fig:wagglyPath}:
    \begin{itemize}
        \item A  directed path from $u^*$ to $u$;
        \item The arc $(u,v_1)$;
        \item For each odd $i \in \{1, \dots, s-2\}$, a directed path from $v_i$ to $u_{i+1}$; % (such a path exists by the observation above);
        \item For each even $i \in \{1,\dots, s-1\}$, the arc $(u_i,v_{i+1})$;
        \item If $s$ is even, a directed path from $v_{s-1}$ to $v_s$.
    \end{itemize}

\begin{figure}[h]
    \centering
    \includegraphics{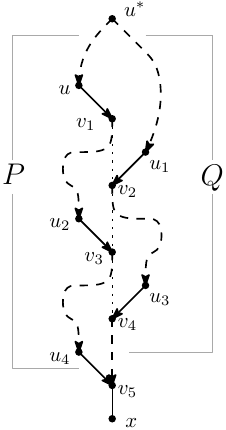}
    \caption{Illustration of the proof of Lemma~\ref{lem:wavyPaths} for the case that $s = 5$. For each $i < s$, $v_{i+1}$ is a lowest vertex below $v_i$ such that $v_{i+1}$ has a parent $u_i$ that is not below $v_i$.
    Dashed lines represent directed paths.
    The path $P$ is in bold on the left, starting at $u^*$, passing through $u,v_1,u_2,v_3,u_4$ and ending at $v_5$.
    The path $Q$ is in bold on the right, starting at $u^*$, passing through $u_1,v_2,u_3,v_4$ and ending at $v_5$. \mj{A dotted line from $v_i$ to $v_{i+1}$ illustrates the fact that $v_{i+1}$ is below $v_i$, for $i \in \{1,\dots, s-2\}$.}}
    \label{fig:wagglyPath}
\end{figure}

    We now have that $P$ contains $v_i$ for all odd $i$ and $u_i$ for all even $i$, and $P$ is a directed path from $u^*$ to $v_s$ (ending with the arc $(u_{s-1},v_s)$ if $s$ is odd, and otherwise ending with an arbitrary path from $v_{s-1}$ to $v_s$). By construction, $P$ contains the arc $(u,v)$.

    In a similar way, we form $Q$ by combining the following directed paths --- see Figure~\ref{fig:wagglyPath}:

    \begin{itemize}
        \item An (arbitrary) directed path from $u^*$ to $u_1$;
        \item For each even $i \in \{1, \dots, s-2\}$, an (arbitrary) directed path from $v_i$ to $u_{i+1}$; % (such a path exists by the observation above);
        \item For each odd $i \in \{1,\dots, s-1\}$, the arc $(u_i,v_{i+1})$;
        \item If $s$ is odd, an (arbitrary) directed path from $v_{s-1}$ to $v_s$.
    \end{itemize}

    We now have that $Q$ contains $v_i$ for all even $i$ and $u_i$ for all odd $i$, and $Q$ is a directed path from $u^*$ to $v_s$  (ending with the arc $(u_{s-1},v_s)$ if $s$ is even, and otherwise ending with an arbitrary directed path from $v_{s-1}$ to $v_s$).
    
    Letting $r$ be the leaf-reticulation $v_s$, we have that $P$ and $Q$ are paths from $u^*$ to $r$.
    It remains to show that $P$ and $Q$ are arc-disjoint. 
    For this, it is sufficient to show that there is no vertex $v'$ in both $P$ and $Q$ except for the $u^*$ and $v_s$. We note that the degenerate case that $P$ and $Q$ both consist of the single arc $(u^*,v_s)$ cannot occur, since we assumed $v_1$ is not a leaf-reticulation and so $s>1$.
    
    So suppose for a contradiction that such a vertex $v'$ does exist. Then $v'$ is strictly below $u^*$ and strictly above $v_s$.
    
    \leo{First suppose} that $v'$ is strictly above $v_1$, \mj{and therefore $v_2$}. %by considering $P$,
    \leo{Since~$v'$ is on~$P$}, this implies that $v'$ is \mj{also} above $u$.
    %As $v'$ must also be strictly above $v_2$, by considering $Q$ we have that $v'$ is above $u_1$. 
    % \leo{In addition, it implies that~$v'$ is strictly above~$v_2$.
    \leo{Since~$v'$ is on~$Q$, it follows that~$v'$ is above~$u_1$.} Thus $v'$ is a common ancestor of $u$ and $u_1$ that is strictly below $u^*$, a contradiction by the choice of $u^*$.
    
    \leo{Now suppose} that $v'$ is below $v_1$. Let $i \in \{1,\dots, v_{s-1}\}$ be the unique index such that $v'$ is below $v_i$ but not below $v_{i+1}$. Since one of the paths $P$ and $Q$ contains $(u_i, v_{i+1})$, $v'$ must be above $u_i$. But then we have that there is a directed path from $v_i$ to $u_i$ via $v'$. \mj{Thus $u_i$ is below $v_i$,} a contradiction by the choice of $u_i$.
    Thus we may conclude that $P$ and $Q$ have no vertices in common except for \mj{ $u^*$ and $v_s$} \leo{(and do not both consist of a single arc)}, and so $P$ and $Q$ are arc-disjoint.
\end{proof}

%\vm{maybe just define these concepts here rather than repeating? Leo: we want to use them also in other sections}
%Recall that a reticulation~$r$  is a \emph{sink} of a cycle~$C$ \mj{in a network} if~$C$ uses both in-arcs of~$r$, that a cycle $C$ is \emph{good} if it has exactly one sink and a good cycle is \emph{excellent} if its sink is adjacent to a leaf.
% A \emph{cycle} in a network~$N$ is a sequence $(v_1,e_1,v_2,e_2\ldots,v_p=v_1)$, $p \ge 4$, alternating between vertices $v_i$ and edges or arcs~$e_j$ such that~$v_i\neq v_j$ for $1\leq i<j<p$ and for all $i\in\{1,\ldots ,p-1\}$ either $e_i=(v_i,v_{i+1})$ or $e_i=(v_{i+1},v_{i})$ is an arc of~$N$ or $e_i=\{v_i,v_{i+1}\}$ is an edge of~$N$. \revnew{We may also describe a cycle by only its vertices~$(v_1,v_2,\ldots,v_p=v_1)$.}
\revnew{We say that two cycles in $N$ \emph{overlap} if they have at least one vertex in common. 
Since~$N$ is binary, two cycles in~$N$ overlap if and only if they have at least one edge or arc in common. Recall that a reticulation~$r$ in $N$ is a \emph{sink} of a cycle~$C$ if~$C$ contains both incoming arcs of~$r$. We call
%Call 
a cycle $C$ \emph{good} if it contains exactly one sink, and we call
% Call 
a good cycle \emph{excellent} if its sink is adjacent to a leaf. \mj{See Figure~\ref{fig:cyclesExamples}.}}

\begin{figure}
\centerline{\includegraphics{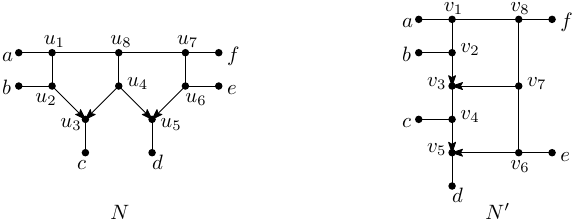}}
% \begin{subfigure}{0.5\textwidth}
%     \includegraphics[scale = 0.8]{CyclesExample1.eps}  
%     \caption{\label{fig:cyclesExample1} $N$}
% \end{subfigure}
% \begin{subfigure}{0.5\textwidth}
%     \includegraphics[scale = 0.8]{CyclesExample2.eps}  
%     \caption{\label{fig:cycleExample2} $N'$}
% \end{subfigure}
    \caption{\label{fig:cyclesExamples} \leon{Two semi-directed phylogenetic networks~$N$ and~$N'$, each containing three pairwise-overlapping cycles.
    %\mj{In~$N$, let $C_1$ denote the cycle with vertices $\{u_1, u_2, u_3, u_4, u_8\}$, $C_2$  the cycle with vertices $\{u_4, u_5,  u_6,  u_7,  u_8\}$, and $C_3$ the cycle with vertices $\{u_1, u_2, u_3, u_4, u_5, u_6, u_7, u_8\}$. These three cycles pairwise overlap.
     In~$N$, the cycle %with vertices
     \revnew{$(u_1, u_2, u_3, u_4, u_8, u_1)$} is good as it has a single sink $u_3$, and it is excellent as $u_3$ is adjacent to a leaf. Similarly, the cycle %with vertices
     \revnew{$(u_4, u_5,  u_6,  u_7,  u_8, u_4)$} is excellent. However, the cycle %with vertices
     \revnew{$(u_1, u_2, u_3, u_4, u_5, u_6, u_7, u_8,u_1)$} is not good (and therefore not excellent) as it has two sinks $u_3$ and $u_5$.
     %In~$N'$, let $C_4$ denote the cycle with vertices $\{v_1,v_2,v_3,v_7,v_8\}$, $C_5$ the cycle with vertices $\{v_3,v_4,v_5,v_6,v_7\}$, and $C_6$ the cycle with vertices $ \{v_1,v_2,v_3,v_4,v_5,v_6,v_7,v_8\}$.
     In~$N'$, the cycle %with vertices
     \revnew{$(v_1,v_2,v_3,v_7,v_8, v_1)$} is good as it has a single sink $v_3$, but it is not excellent since $v_3$ is not adjacent to a leaf.
     The cycles %with vertices 
     \revnew{$(v_3,v_4,v_5,v_6,v_7,v_3)$}, and \revnew{$(v_1,v_2,v_3,v_4,v_5,v_6,v_7,v_8,v_1)$} are both excellent.
     Note in particular that this last cycle is good even though it contains two reticulations $v_3$ and $v_5$, as $v_3$ is not a sink of this cycle.}
     }
\end{figure}

We say a leaf \emph{belongs} to a \leo{cycle} $C$ if the unique vertex that is adjacent to it is in $C$. Note that if~$r$ is the sink of a \leo{good} cycle~$C$ in a semi-directed network~$N$ and $x$ is \leo{a leaf} below $r$, then $x$ belongs to $C$ if and only if $r$ and $x$ are adjacent. To see this, suppose that~$x$ is a leaf below~$r$ and belongs to~$C$ but is not adjacent to~$r$. Then there exists a semi-directed path from~$r$ to the 
vertex~$v$ adjacent to~$x$. Since~$v$ \mj{is a vertex of} $C$ and~$r$ is the unique sink of~$C$, there exists a semi-directed path from~$v$ to~$r$. Hence, there exists a semi-directed cycle in~$N$, which is a contradiction by Lemma~\ref{lem:sdcycle}.

\begin{lemma}\label{lem:excellentCycle}
    Let $N$ be a simple, semi-directed \leon{phylogenetic} network \mj{with at least one reticulation} and let $e$ be an arc or edge between two non-leaf vertices.
    Then $e$ is contained in at least one excellent cycle. %\todo{seems true for $k=1$ as well}
\end{lemma}
\begin{proof}

    Let $v_1,v_2$ be the vertices of $e$.
    Let $N_d$ be \leon{a rooting of $N$} \mj{with root $\rho$}.
    %such that $V(N_d) = V(N)\cup \{\rho\}$ where $\rho$ is the root of $N_d$ (which exists by Lemma~\ref{lem:suppressroot}).
    Observe that either $v_1$ and $v_2$ are adjacent in $N_d$, or $N_d$ contains the arcs $(\rho, v_1), (\rho, v_2)$ (and $\rho$ is not adjacent to any other vertices).
    If $v_1$ and $v_2$ are adjacent in $N_d$, we may assume without loss of generality that the arc between them is $(v_1,v_2)$.

    Now, let $(u,v) = (v_1,v_2)$ if $v_1$ and $v_2$ are adjacent in $N_d$, and let $(u,v) = (\rho, v_1)$ otherwise.
    By Lemma~\ref{lem:wavyPaths} there exist arc-disjoint directed paths $P,Q$ in $N_d$ from some vertex $u^*$ to a leaf-reticulation $r$, and $(u,v)$ is on the path $P$.
    Note that either $u^* = \rho$ or every vertex in $P$ and $Q$ is a vertex of $N$.

    We now construct a cycle $C$ in $N$ from the union of $P$ and $Q$. For each arc $e'$ in $P$ or $Q$ not incident to $\rho$, let $e''$ be the corresponding edge or arc in $N$ (i.e. with the same vertices as $e'$), and add $e''$ to $C$.
    If $u^* = \rho$, then \leon{$(\rho,v_1)$ and $(\rho, v_2)$ are} the first arcs of $P$ and $Q$ %respectively
    %\vm{do we have $\{u_1,u_2\}=\{v_1,v_2\}$?} \leon{I think so}
    and add the arc or edge in $N$ between \leon{$v_1$ and $v_2$} to $C$. %\todo{Is it clear this must exist? Yes}
    Since $P$ and $Q$ are arc-disjoint paths with the same start and end vertices, the resulting $C$ is indeed a cycle. Moreover $C$ contains $e$  (either because $(v_1,v_2)$ is an arc in $P$, or because $(\rho, v_1)$ and $(\rho, v_2)$ are the top arcs of $P$ and $Q$ respectively). It remains to show that $C$ is an excellent cycle.

    To see that $C$ is a good cycle, observe that any sink in $C$ must have two incoming arcs in the union of $P$ and $Q$. But as $P$ and $Q$ are \revnew{edge-disjoint} directed paths in $N_d$ there is only one vertex for which this holds, namely~$r$. Thus $C$ has only one sink.
    Finally, as $r$ is a leaf-reticulation, there is a leaf \leo{adjacent to $r$} %belonging to $C$,
    and so $C$ is excellent.
\end{proof}

\begin{lemma}\label{lem:twogood} Let~$N$ be a simple, semi-directed \leon{phylogenetic} network on~$X$ with \mj{at least one reticulation}, and let $A|B$ be any bipartition of~$X$. Then there exist excellent cycles $C_1, C_2$ (possibly with $C_1=C_2$) and leaves $a \in A, b \in B$ such that $C_1$ and $C_2$ overlap and $a$ belongs to $C_1$ and $b$ belongs to $C_2$. \leo{In addition, either $C_1\neq C_2$ and~$a$ and~$b$ are both adjacent to a reticulation or~$C_1= C_2$ and one of~$a$ and~$b$ is adjacent to a reticulation.}
\end{lemma}
\begin{proof}
    Take any $a' \in A$ and $b' \in B$. Let $v_a$ be the non-leaf vertex adjacent to $a'$ and $v_b$ the non-leaf vertex adjacent to $b'$. Since $N$ is connected, there exists a path (not necessarily semi-directed) between $v_a$ and $v_b$, and all vertices on this path are non-leaf vertices.
    Let $v_1 = v_a,v_2,\dots, v_s = v_b$ be the vertices of this path, and let $e_i$ be the edge or arc between $v_i$ and $v_{i+1}$, for each $i \in \{1,\dots, s-1\}$.
    By Lemma~\ref{lem:excellentCycle}, for each  $i \in \{1,\dots, s-1\}$ there exists an excellent cycle $C_i'$ containing $e_i$.
    As each $C_i'$ is an excellent cycle, it has at least one leaf in $A$ or $B$ belonging to it (namely the leaf adjacent to its sink).    
    Note that in particular $a'$ belongs to $C_1'$ since $C_1'$ contains $v_a$, and $b'$ belongs to $C_{s-1}'$ since $C_{s-1}'$ contains $v_b$.  
    Therefore there exists some $i \in \{1,\dots, s-2\}$ such that a leaf $a$ in $A$ belongs to $C_i'$, and a leaf $b$ in $B$ belongs to $C_{i+1}'$. Furthermore $C_i'$ and $C_{i+1}'$  must overlap, as they both contain the vertex $v_{i+1}$.
    Then $C_i'$ and $C_{i+1}'$ are the desired excellent cycles.
    
    \leo{Finally, note that we can choose~$a$ and~$b$ to be both adjacent to a reticulation unless the leaves adjacent to the sinks of the $C_i'$ are all in~$A$ or all in~$B$. If they are all in~$A$, then we can take~$C_1=C_2=C_{s-1}'$ and~$a$ is adjacent to a reticulation. If the leaves adjacent to the sinks of the $C_i'$ are all in~$B$, then we can take~$C_1=C_2=C_1'$ and~$b$ is adjacent to a reticulation.}
\end{proof}

We are now ready to prove Theorem~\ref{thm:splitsCondition}.

\begin{proof}[Proof of Theorem~\ref{thm:splitsCondition}]

For the first direction of the proof, assume that $A|B$ is a CE-split in~$N$ and $|A|, |B|\geq 2$. Let $a_1,a_2 \in A$, $b_1,b_2 \in B$, all pairwise distinct \mj{and let $Y = \{a_1,a_2,b_1,b_2\}$}.
% Recall that $N|_{\{a_1,a_2,b_1,b_2\}}$ is obtained by applying the suppression operation to~$N$ with respect to~$\{a_1,a_2,b_1,b_2\}$. 
\mj{Recall that $N|_{Y}$ is obtained by applying the suppression operation to $N_{\wedge Y}$ and so the leaf set of  $N|_{Y}$  is therefore $Y$.}
% Observe that operation (V6) deletes all leaves other than $a_1,a_2,b_1,b_2$. 
Moreover, \mj{and CE-split in $N$ is also a CE-split in $N_{\wedge Y}$, and } all suppression operations preserve CE-splits (but not the number of corresponding cut-edges). Hence, $\{a_1,a_2\}|\{b_1,b_2\}$ is a CE-split in $N|_{\{a_1,a_2,b_1,b_2\}}$.

To see the reverse direction, we use  induction on the number of non-trivial CE-splits in~$N$. The base case is that~$N$ is simple. To see this case, note that 
\mj{if $A|B$ is a trivial split of $X$ then it is certainly a CE-split in $N$.
So assume that $A|B$ is not a trivial split of $X$. We claim that if $A|B$ is not a CE-split in $N$ then 
	there exist $a_1, a_2 \in A$, $b_1,b_2 \in B$ such that $\{a_1,a_2\}|\{b_1,b_2\}$ is not a CE-split in $N|_{\{a_1,a_2,b_1,b_2\}}$. 
	By contraposition, this completes the proof of this direction for the base case.}
% \leon{In order to derive a contradiction, 

\kh{To see the claim, assume}
\leon{%Assume 
	that $A|B$ is not a CE-split in $N$.}
By Lemma~\ref{lem:twogood}, there exist \leo{excellent} cycles $C_1, C_2$ and leaves $a_1 \in A, b_1 \in B$ such that $C_1, C_2$ overlap, $a_1$ belongs to $C_1$ and $b_1$ belongs to $C_2$. In addition, either $C_1\neq C_2$ and~$a_1$ and~$b_1$ are both adjacent to a reticulation or~$C_1= C_2$ and one of~$a_1$ and~$b_1$ is adjacent to a reticulation.

Let~$a_2$ be an arbitrary element of~$A\setminus\{a_1\}$ and let~$b_2$ be an arbitrary element of~$B\setminus\{b_1\}$ \kh{which must exist because $A|B$ is not a trivial split of $X$. Put $Y=\{a_1,a_2,b_1,b_2\}$}. We claim that $\{a_1,a_2\}|\{b_1,b_2\}$ is not a CE-split in $N|_Y$. To see this, \kh{note first that since $C_1$ and $C_2$ are excellent,}
% that~$C_1$ and~$C_2$ 
both \kh{of them must} have a unique sink adjacent to a leaf \kh{in~$\{a_1,b_1\}$.}
%in~$\{a_1,a_2,b_1,b_2\}$.
Hence, these sinks and adjacent leaves are not deleted \kh{when constructing $N|_{\wedge Y}$ from $N$.}
% suppression operations (V5) and (V6). 
Moreover, if~$C_1\neq C_2$, then the cycles~$C_1$ and~$C_2$ are not suppressed by operation (PAS) \kh{to obtain $N|_Y$} because each of them has at least three 
\kh{vertices that each \revnew{is} incident with edges/arcs that do not form part of the cycle}
%incident edges/arcs 
(one incident to~$a_1$ or~$b_1$ and two belonging to the other cycle). If~$C_1=C_2$, this cycle is also not suppressed by operation (PAS) \kh{to obtain $N|_Y$} because the cycle contains at least 
\kh{three vertices that each \revnew{is} incident with edges/arcs that do not form part of the cycle}
%three incident edges/arcs 
(one incident to~$a_1$, one incident to~$b_1$ and one incident to or on a path to~$a_2$). Finally, the blob containing~$C_1$ and~$C_2$ is not suppressed by operation (BLS) by the same argument. Hence, \kh{although the length of} the cycles~$C_1$ and~$C_2$ may be shortened due to \kh{applied} suppression operations, they still exist (and still overlap) in $N|_{\{a_1,a_2,b_1,b_2\}}$. Since~$a_1$ belongs to~$C_1$ and~$b_1$ belongs to~$C_2$, it follows that $\{a_1,a_2\}|\{b_1,b_2\}$ is not a CE-split in $N|_{\{a_1,a_2,b_1,b_2\}}$, \kh{as claimed}.

\kh{Assume that the theorem holds for all semi-directed phylogenetic networks $N'$ on $X$ and all bipartitions of $X$ if $N'$ has strictly less CE-splits than $N$ and that}
there exists a non-trivial CE-split $P|Q$ in~$N$.

\kh{We first show that $P\subseteq A$, \kh{or} $P\subseteq B$, \kh{or} $Q\subseteq A$ or $Q\subseteq B$ must hold. Assume for contradiction that this is not the case.
Then}~$P$ contains leaves~$a_1\in A$ and~$b_1\in B$ and $Q$ contains leaves~$a_2\in A$ and~$b_2\in B$. Then $\{a_1,b_1\}|\{a_2,b_2\}$ is a CE-split in $N|_{\{a_1,a_2,b_1,b_2\}}$. Hence, $\{a_1,a_2\}|\{b_1,b_2\}$ is not a CE-split in $N|_{\{a_1,a_2,b_1,b_2\}}$, a contradiction.

Hence, we have that $P\subseteq A$, \kh{or} $P\subseteq B$, \kh{or} $Q\subseteq A$ or $Q\subseteq B$. Without loss of generality, assume that~$P\subseteq A$. Let~$\{u,v\}$ be a cut-edge such that 
%all 
deleting~$\{u,v\}$ creates two connected components: one connected component~$N_P$ containing~$u$ and all leaves from~$P$ and one connected component containing~$v$ and all leaves from~$Q$. Construct \kh{a} network~$N'$ from~$N_P$ by adding a new leaf~$a^*$ and an edge $\{a^*,u\}$. Let~$A'=(A\setminus P)\cup\{a^*\}$ and~$B'=B$. Note that~$N'$ has at least one non-trivial CE-split less than~$N$. To be able to apply induction \kh{to $N'$}, we need that, for any $a_1,a_2 \in A', b_1,b_2 \in B'$, $\{a_1,a_2\}|\{b_1,b_2\}$ is a CE-split in $N'|_{\{a_1,a_2,b_1,b_2\}}$. If~$a^*\notin\{a_1,a_2\}$ then this is clear because $\{a_1,a_2\}|\{b_1,b_2\}$ is a CE-split in $N|_{\{a_1,a_2,b_1,b_2\}}$ and hence also in $N'|_{\{a_1,a_2,b_1,b_2\}}$. If~$a^*\in\{a_1,a_2\}$ then assume without loss of generality that $a^*=a_1$. Let~$c\in P$. Then $\{c,a_2\}|\{b_1,b_2\}$ is a CE-split in $N|_{\{c,a_2,b_1,b_2\}}$ and hence $\{a_1,a_2\}|\{b_1,b_2\}$ is a CE-split in $N'|_{\{a_1,a_2,b_1,b_2\}}$. Hence, by induction, $A'|B'$ is a CE-split in~$N'$. It follows directly that~$A|B$ is a CE-split in~$N$.
\end{proof}

We conclude this section by noting that, since undirected phylogenetic trees are encoded by their splits, it follows from Theorem~\ref{thm:splitsCondition} that the blob tree of a semi-directed phylogenetic network is uniquely determined by the quarnets of the network. Stated more precisely, we have the following corollary.

\begin{corollary}\label{cor:blobtree-encoding}
Suppose that~$N$ and~$N'$ are semi-directed phylogenetic networks on~$X$ with~$Q(N)\simeq Q(N')$. 
Then $B(N)\cong B(N')$.
\end{corollary}

\section{Level-2 networks}\label{sec:lev2combined}

In this section, we combine the results from Sections~\ref{sec:lev2} and~\ref{sec:cuts} to prove that semi-directed level-2 networks with at least four leaves 
are encoded by their quarnets. For that, we will need the following lemma.

\begin{lemma}\label{lem:threecycles}
    Let $N$ be a semi-directed, strict level-$k$ \leon{phylogenetic} network on $X$, $|X|\ge 4$, for $k \geq 3$. Then there exists a quarnet~$q\in Q(N)$ such that~$q$ is not level-2.
\end{lemma}
\begin{proof}

\begin{figure}
    \centerline{\includegraphics{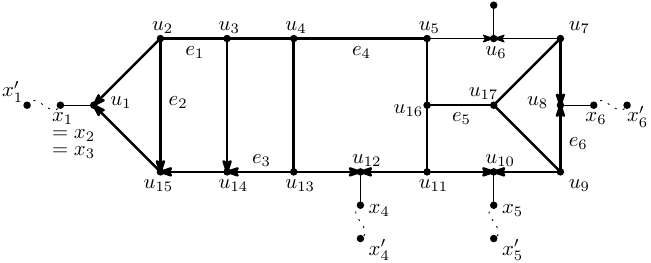}}
    \caption{\label{fig:threecycles} Illustration used in the proof of Lemma~\ref{lem:threecycles}. The solid edges indicate the network~$N_B$. The dotted edges indicate paths outside~$N_B$ to leaves of~$N$. Bold edges indicate the final~$U$.
    %{\color{blue} hard to see the bold - make thicker?}
    % LEO: done
    The indicated edges~$e_1,\ldots ,e_6$ are one possibility for the edges chosen in the proof of Lemma~\ref{lem:threecycles}. In that case,~$C_1,\ldots ,C_6$ could be the excellent cycles with vertices $(u_1,u_2,u_3,u_{14},u_{15},u_1)$, $(u_1,u_2,u_{15},u_1)$, $(u_1,u_2,u_3,u_4,u_{13},u_{14},u_{15},u_1)$, $(u_4,u_5,u_{16},u_{11},u_{12},u_{13},u_4)$, $(u_9,u_{10},u_{11},u_{16},u_{17},u_9)$ and $(u_7,u_8,u_9,u_{17},u_7)$ respectively. This leads to the quarnet~$q=N|_{\{x'_1,x'_4,x'_5,x'_6\}}$, which is not level-$2$.}
\end{figure}

Consider any blob~$B$ of~$N$ with exactly $k$ reticulations. Let~$N_B$ be the semi-directed, simple, strict level-$k$  network obtained from~$N$ by deleting all vertices that are not in~$B$ and do not have an adjacent vertex that is in~$B$.

We construct a set~$A\subseteq L(N_B)$ %with~$|A|=4$
and a set~$\cC$ of excellent cycles with~$|\cC|\geq 3$ in~$N_B$ such that each~$C\in\cC$ overlaps with at least one~$C'\in\cC\setminus \{C\}$ as follows. See Figure~\ref{fig:threecycles} for an example.

Let~$e_1$ be any edge/arc of~$N_B$ between two non-leaf vertices. Then, by Lemma~\ref{lem:excellentCycle}, there exists an excellent cycle~$C_1$ in~$N_B$ containing~$e_1$. Let~$x_1$ be the leaf of~$N_B$ below the sink of~$C_1$. Initialize~$A=\{x_1\}$,~$\cC=\{C_1\}$ and~$U=C_1$.

Repeat the following while~$U\neq B$ and~$|A|<4$. Let~$i=|\cC|+1$ and~$e_i$ any edge/arc between two non-leaf vertices of~$N_B$, such that~$e_i$ is not in~$U$ but is incident to at least one vertex in~$U$. 
Note that~$e_i$ exists since~$U\neq B$. By Lemma~\ref{lem:excellentCycle}, there exists an excellent cycle~$C_i$ in~$N_B$ containing~$e_i$. Note that~$C_i\neq C$ for all~$C\in\cC$ and that~$C_i$ overlaps with at least one~$C\in\cC$. Let~$x_i$ be the leaf of~$N_B$ below the sink of~$C_i$. Add~$C_i$ to~$\cC$, add~$x_i$ to~$A$ (note that~$x_i$ may already be in~$A$, in which case~$A$ remains unchanged) and update~$U$ to be the graph union of the cycles in~$\cC$. 

First suppose~$|A|=4$. In this case, we have~$|\cC|\geq 4$ and hence~$U$ is not level-2. 
(To see this, note that~$C_1\in\cC$ contains a leaf reticulation,~$C_2\in\cC\setminus\{C_1\}$ either contains a different leaf reticulation or it joins~$C_1$ in a different reticulation and~$C_3\in\cC\setminus\{C_1,C_2\}$ either has a leaf reticulation that is different from the leaf reticulations of~$C_1$ and~$C_2$ or it joins~$C_1\cup C_2$ in a third reticulation.)
%{\color{blue} Is this obvious?}
Consider the quarnet~$q_B=N_B|_A$. We now show that~$q_B$ is not level-2. To see this, first recall that each~$C\in \cC$ has a unique sink with a leaf in~$A$ below it and sinks are not deleted by vertex suppression operations. Moreover, none of the cycles~$C\in\cC$ can be suppressed by operation (PAS). To see this, recall that~$C$ corresponds to an excellent cycle in~$N_B$ and hence its sink is incident to a cut-edge in~$N_B$. Moreover, since~$C$ overlaps with at least one~$C'\in\cC\setminus\{C\}$, it either has a chord (i.e., an edge/arc that is not in~$C$ but is incident to two vertices of~$C$) or three incident edge/arcs (one where~$C'$ leaves~$C$, one where~$C'$ joins~$C$ again, and one incident to the sink of~$C$). In either case,~$C$ is not suppressed by (PAS).
Finally, the blob suppression operation (BLS) is not applicable to~$U$ because it has at least four incident cut-edges (incident to the leaves in~$A$). Hence~$q_B$ is not level-$2$. Let~$A'\subseteq X$ consist of, for each~$x_i\in A$, one leaf~$x'_i$ of~$N$ that is below~$x_i$ in~$N$. Then~$q=N|_{A'}$ is equal to~$q_B$ with each leaf~$x_i$ replaced by~$x_i'$. Hence,~$q$ is not level-2.

Now consider the case that~$|A|<4$. In this case we have $U=B$ because otherwise the while loop would not have terminated. Let~$A'\subseteq X$ contain, for each~$x_i\in A$, one leaf~$x'_i$ of~$N$ that is below~$x_i$ in~$N$. In addition, add arbitrary leaves from~$X$ to~$A'$ until~$|A'|=4$. Then~$q=N|_{A'}$ contains a blob~$U=B$ in which no suppression operations are applicable since~$B$ is a blob of~$N$ which is phylogenetic. Hence,~$q$ is not level-2 since it contains~$B$ which is not level-2.
\end{proof}

We are now ready to prove the main result of this section.

\begin{theorem}\label{the:class-lev2}
The class of semi-directed, level-2, \leon{binary phylogenetic} networks with at least four leaves is encoded by quarnets.
\end{theorem}
\begin{proof}
Let~$N$ be a semi-directed level-2 \leon{phylogenetic} network with at least four leaves. Let $X$ be the leaf set of $N$. Let~$N'$ be a semi-directed network on~$X$ with~$Q(N)\simeq Q(N')$. We need to show that~$N\cong N'$.

First we prove that~$N'$ has level-2. Assume for a contradiction that~$N'$ is strict level-$k$ with $k\geq 3$. By Lemma~\ref{lem:threecycles}, there exists a quarnet~$q\in Q(N')$ that is not level-$2$. This leads to a contradiction since~$q\in Q(N) \simeq Q(N')$ and~$N$ has level-2. \kh{Thus, $N'$ is a level-2 network.}

We now prove that $N\cong N'$ by induction on the number~$s$ of nontrivial CE-splits in~$N$. 

If~$s=0$, then~$N$ is a semi-directed, simple level-2 phylogenetic network on~$X$. \kh{Since $Q(N)\simeq Q(N')$ it follows that $N'$ is also a semi-directed, simple, level-2 phylogenetic network.} By Theorem~\ref{thm:lev2}, $N\cong N'$ \kh{follows}.

\kh{So assume that} $s\geq 1$. \kh{Observe that, by Theorem~\ref{thm:splitsCondition},~$N'$ has the same CE-splits as~$N$.} Consider a nontrivial CE-split~$A|B$ of~$N$ and~$N'$ (which exists since~$s\geq 1$).
%{\color{blue} need some remark saying why we can assume this exists?} 
Pick some~$a\in A$ and~$b\in B$ and consider the networks $N|_{A\cup\{b\}}$ and $N|_{B\cup\{a\}}$. Since $
Q(N|_{A\cup\{b\}}) \simeq \{q\in Q(N) \mid L(q)\subseteq A\cup\{b\}\}$,
$Q(N'|_{A\cup\{b\}}) \simeq \{q\in Q(N') \mid L(q)\subseteq A\cup\{b\}\}$,
and  $Q(N)\simeq Q(N')$,
we have that $Q(N|_{A\cup\{b\}})\simeq Q(N'|_{A\cup\{b\}})$. If we also have $|A\cup\{b\}|\geq 4$ then it follows by induction that $N|_{A\cup\{b\}}\cong N'|_{A\cup\{b\}}$. Otherwise, we have $|A|=2$ and there exists $b'\in B$ with~$b'\neq b$.
%Using similar arguments as above, $Q(N|_{A\cup\{b,b'\}})\simeq Q(N'|_{A\cup\{b,b'\}})$.
It then follows directly that $N|_{A\cup\{b,b'\}}\cong N'|_{A\cup\{b,b'\}}$ (since both are quarnets) and hence that $N|_{A\cup\{b\}}\cong N'|_{A\cup\{b\}}$ (since both can be obtained from $N|_{A\cup\{b,b'\}}$ by deleting~$b'$ and applying the suppression operation).
%{\color{blue} why is this following?}.
% is it clear now?

By symmetry, we also have that $N|_{B\cup\{a\}}\cong N'|_{B\cup\{a\}}$.

Since $A|B$ is a CE-split, there exists a cut-edge $\{u,v\}$ of $N$ such that the removal of it results in two connected graphs~$N_A,N_B$ with leaf sets $A$ and $B$, respectively. Without loss of generality, $u$ is in~$N_A$ and $v$ is in~$N_B$. Observe that, by definition, $N|_{A\cup\{b\}}$ can be obtained from~$N_A$ by adding leaf~$b$ with an edge~$\{u,b\}$.
% {\color{blue} is this obvious?}
% LEO: it follows directly from the way both subnetworks are defined
Similarly, $N|_{B\cup\{a\}}$ can be obtained from~$N_B$ by adding leaf~$a$ with an edge~$\{v,a\}$. 
% Now, let~$n_b$ be the vertex adjacent to~$b$ in $N|_{A\cup\{b\}}$ and let~$n_a$ be the vertex adjacent to~$a$ in $N|_{B\cup\{a\}}$. 
Then,~$N$ can be obtained from~$N|_{A\cup\{b\}}$ and $N|_{B\cup\{a\}}$ by deleting~$b$ and its incident edge from~$N|_{A\cup\{b\}}$, deleting~$a$ and its incident edge from~$N|_{B\cup\{a\}}$ and adding an edge $\{u,v\}$. In exactly the same way,~$N'$ can be obtained from~$N'|_{A\cup\{b\}}$ and $N'|_{B\cup\{a\}}$. Since $N|_{A\cup\{b\}}\cong N'|_{A\cup\{b\}}$ and $N|_{B\cup\{a\}}\cong N'|_{B\cup\{a\}}$, it follows that $N\cong N'$.
\end{proof}

\section{Discussion}\label{sec:discussion}

In this paper we have shown that the set of quarnets of a semi-directed level-2 phylogenetic network encodes the network, but that  this is no longer necessarily true
for level-3 networks. \rev{In addition, we proved that the blob tree of a semi-directed phylogenetic network is encoded by the quarnets of the network for any level.} % (as proven in the introduction).

There are several directions that could be of interest to be investigated next. First, it could be useful for practical applications to develop algorithms
that compute semi-directed level-2 networks from collections of quarnets. As a first
step in this direction it would be interesting to develop an algorithm that 
computes a semi-directed level-2 network from its set of quarnets \rev{(see~\cite{frohn2024reconstructing} for such an algorithm for level-$1$). We could then adapt the algorithm to robustly deal with arbitrary collections of level-2 quarnets, similar to the \textsc{Squirrel} and \textsc{NANUQ+} algorithms for 
level-$1$~\cite{allman2024nanuq+,holtgrefe2024squirrel}.

An $O(n^3)$-time algorithm for constructing the blob tree of a semi-directed phylogenetic network of any level from quarnets was recently developed~\cite{frohn2024reconstructing} based on the results in this paper.} \revnew{An interesting open problem is whether the blob tree can be reconstructed from only $O(n^2)$ quarnets and whether this is possible in $O(n^2)$ time. From a practical point-of-view it is important to develop robust blob tree construction methods. If~$n$ is not too big, practical algorithms could use information from all $O(n^4)$ quarnets~\cite{allman2024tinnik,holtgrefe2024squirrel}, but when considering real data such methods currently struggle to decide how resolved to make the blob tree.}

In another direction, it could be interesting to study {\em inference rules} for semi-directed quarnets. For phylogenetic trees, inference rules have been studied for some years, where they are used to infer new trees from collections of trees (see e.g. \cite[Section 6.7]{semple2003phylogenetics}).  In~\cite{huber2018quarnet}, certain inference rules are given for level-1 undirected networks on four leaves, and it would be interesting to see whether similar rules can be developed for the semi-directed case. In a related direction, it could also be worth investigating approaches for deciding whether or not an arbitrary collection of quarnets (i.e. not necessarily one quarnet for each quartet of leaves) can be displayed by some semi-directed phylogenetic network. Note, however, that it is NP-complete to decide whether there is a tree that displays an arbitrary collection of quartet trees~\cite{steel1992complexity}.

Although we have shown that semi-directed level-3 networks are, in general, not encoded by their quarnets, it could be of interest to find a maximal subclass of level-3 (or higher) networks that is encoded by quarnets. In particular, we conjecture that the class of all semi-directed binary \leon{simple} level-3 networks, except for the networks~$N_1,N_2$ in Figure~\ref{fig:lev3} and networks that can be obtained from~$N_1$ and~$N_2$ by inserting leaves on the side of~$a$ and~$b$ (in any order), is encoded by quarnets.

Finally, one major challenge that remains is to develop robust ways to
construct quarnets from real data. This problem
has generated considerable interest in the area of algebraic geometry, 
where the problem of identifying level-1 quarnets using algebraic
invariants arising from models of sequence evolution has yielded some 
positive results on network identifiability (see e.g. \cite{gross2021distinguishing}).
Some recent progress has also been made in \cite{cummings2023computing,martin2023algebraic} for computing 
level-1 quarnets for real data using algebraic invariants, but extending these approaches to level-2 quarnets appears to be a challenging problem~\cite{ardiyansyah2021distinguishing}.

\mj{{\bf Acknowledgements.} We are very grateful to Jannik Schestag for fruitful discussions contributing to the development of this paper \rev{and to Niels Holtgrefe for providing Figure~\ref{fig:squirrel}.}}

{\bf Data availability.} No data was used.

\bibliographystyle{plain} 
\bibliography{bibliography}

\appendix

\section{Omitted proofs}

In this appendix, we provide the previously omitted proofs for Lemmas~\ref{lem:semiDirectedPreserved},~\ref{lem:SuppWellDefinedBrief} and~\ref{lem:wedges}.

%We first give a proof of Lemma~\ref{lem:wedges}.
%\todo{use the original numbering - but not sure how!}
% DONE

\textbf{\Cref{lem:semiDirectedPreserved}.}   
    \emph{Let $N$ be a semi-directed network.
    If $N'$ is derived from $N$ by a single application of (V1), (V2), (BLS) or (PAS), then $N'$ is also a semi-directed network.}
\begin{proof}
    %As no vertex in a semi-directed \leo{network} can have exactly one incoming arc, rule (V3) does not apply to $N$.
    %{\color{blue} so maybe leave (V3) out of the prop statement?}
    % LEO: done
    %So consider any semi-directed network~$N$ and let~$N'$ be obtained from~$N$ by a single application of a suppression operation of type (V1), (V2), (PAS) or (BLS).

If the operation is of type (BLS), then 
\mj{let $B$ be an affected blob in $N$. Note that there is a corresponding blob in any rooting $N_d$ of $N$. If this blob has one incoming and one outgoing arc in $N_d$, then}
the same operation is also applicable to $N_d$ and applying it results in a rooting of~$N'$. 
\mj{Otherwise, the blob corresponding to $B$ has two outgoing arcs and no incoming arcs in $N_d$. Then replacing this blob with a single root vertex again gives a directed network which is a rooting of~$N'$.} 
Hence,~$N'$ is semi-directed.

If the operation is of type (PAS), we claim that there exists a rooting~$N_d$ of~$N$ such that the edge/arc that is subdivided by the root is not one of the suppressed parallel arcs $(u,v)$. To see this, note that by definition of (PAS) vertex~$u$ has degree $3$ and hence has an incident edge~$\{u,w\}$. If there exists a rooting of~$N$ with the root subdividing one of the arcs~$(u,v)$, then there also exists a rooting~$N_d$ of~$N$ with the root subdividing~$\{u,w\}$. Then (PAS) is applicable to~$S_d$ giving a rooting of~$N'$. Hence,~$N'$ is semi-directed.

If the operation is of type (V1) or (V2), we claim that, unless~$N$ has only three vertices, there exists a rooting~$N_d$ of~$N$ such that the edge/arc that is subdivided by the root is not incident to the suppressed vertex~$v$. Let~$\{u,v\}$ be an edge incident to~$v$ such that~$u$ is not a leaf (which exists unless~$N$ has exactly three vertices). Then~$u$ has at least one other incident edge/arc, say to vertex~$p$. If there exists a rooting of~$N$ with the root subdividing one of the edge/arcs incident to~$v$, then there also exists a rooting~$N_d$ of~$N$ with the root subdividing the edge/arc between~$u$ and~$p$. Then suppression operation (V3) is applicable to~$N_d$ giving a rooting of~$N'$. Hence,~$N'$ is semi-directed. Finally, if~$S$ contains exactly three vertices then~$N'$ consists of two vertices connected by an edge and it is clear that~$N'$ is again semi-directed.
\end{proof}

%\subsection{Uniqueness of $\supp{N}$}

\mj{We now turn our attention to the proof of \Cref{lem:SuppWellDefinedBrief}, i.e. that $\supp N$ is well-defined on directed and semi-directed networks.} \leonn{For this, we need some additional definitions and lemmas.}

Let~$N$ be a
%(directed or semi-directed)
network. Call a subgraph $Z$ of $N$ a subgraph \emph{from $u$ to $w$} if all arcs/edges in $Z$ are on some semi-directed path from $u$ to $w$, and $N$ has no arcs/edges incident to $V(Z)\setminus\{u,w\}$ except for those in $E(Z)$.
We call $Int(Z) := V(Z)\setminus\{u,w\}$ the \emph{internal vertices of $Z$}. \leonn{Furthermore, if~$A$ is a subset of the vertices of~$N$, then $N[A]$ is used to denote the subnetwork of~$N$ induced by~$A$, i.e., the subnetwork obtained by deleting all vertices not in~$A$.}

We now characterize which vertices will be suppressed by the suppression operation in a directed network. We will show that these are precisely the internal vertices of subgraphs of the following type. Define the \emph{\mj{directed} SP-graphs (suppressed graphs)} as follows:

\begin{itemize}
    \item (single arc) the graph $(V = \{u,w\}, E =\{(u,w)\})$ is a directed SP-graph from $u$ to $w$.
    % \item (single edge) the graph $(V = \{u,w\}, E =\{\{u,w\}\})$ is an SP-graph from $u$ to $w$.
    % \item (degree-2) the graph $(V = \{u,v,w\}, E=\{(u,v),(v,w)\})$ is an SP-graph.
    % \item (degree-2) if $v$ is a degree-2 vertex with neighbors $u,w$, and $e_1$ is the arc $(u,v)$ or edge $\{u,v\}$ and $e_2$ the arc $(v,w)$ or edge $\{v,w\}$, then the graph $(V = \{u,v,w\}, E=\{e_1,e_2\})$ is an SP-graph.
    \item (parallel arcs) the graph with $V = \{u,v\}$ and parallel arcs $(u,v)$ is a directed SP-graph from~$u$ to~$v$. 
    \item (series) If $Z_1 = (V_1,E_1)$ is a directed SP-graph from $u$ to $v$ and $Z_2 = (V_2,E_2)$ is a  SP-graph from $v$ to $w$ with $V_1\cap V_2 = \{v\}$, then $(V_1\cup V_2, E_1 \cup E_2)$ is a SP-graph from $u$ to $w$. 
    % \item (parallel arcs \leo{plus an arc}) the graph $(V = \{u,v,w\}, E = \{(u,v), (u,v), (v,w))$ is an SP-graph.\todo{can we let just parallel arcs be an SP-graph?}
    \item (recursion) If $Z$ is a directed SP-graph from $u$ to $w$ and $Z'$ an SP-subgraph of $Z$ from $u'$ to $w'$
    (i.e. $Z'$ is a subgraph of $Z$ that is a directed SP-graph), where $u'$ appears before $w'$ in a 
    %semi-
    \mj{directed path} from $u$ to $w$,
    then the result of replacing $Z'$ with another directed SP-graph from $u'$ to
    $w'$ is also a directed SP-graph from $u$ to $w$.    
 \end{itemize}

\begin{figure}
    \centering
    \includegraphics{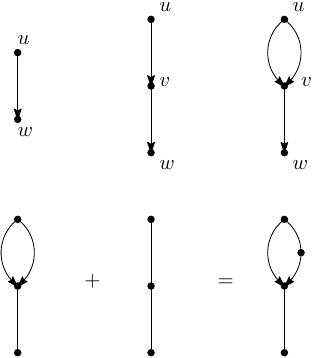}
    \caption{Examples of SP-graphs.}
    \label{fig:sp-graphs}
\end{figure}
%\todo{(In Fig \ref{fig:sp-graphs}) Reticulations in the lower graphs should have arrow heads on the incoming arcs}

We note without proof the following properties of a \mj{directed} SP-graph~$Z$: 
%a \mj{directed} SP-graph
$Z$ has a single vertex $u$ of indegree $0$ and a single vertex $w$ of outdegree $0$, and all arcs of $Z$ are on a directed path from $u$ to $w$. Every \mj{directed} SP-graph with more than one arc has either a vertex of degree $2$, or a pair of parallel arcs.

\mj{Define the \emph{semi-directed SP-graphs} to be the mixed graphs that can be derived from a directed SP-graph by unorienting all arcs except for those entering reticulations. We call a mixed graph an \emph{SP-graph} if it is a directed or semi-directed SP-graph.}

 When $Z$ is an SP-subgraph from $\rho$ to $w$ in $N$ for $\rho$ the root of $N$, and $Z$ contains both out-arcs of $\rho$, then we say $Z$ is \emph{degenerate}.
%In particular, pair of parallel arcs $(\rho, v)$ form a degenerate SP-subgraph together with the child of $v$.

 \begin{lemma}\label{lem:SPsubgraphs}
     Let $N_1,N_2$ be networks such that $N_2$ is derived from $N_1$ by an application of (V1) or (V2) \mj{or (V3)}  or (PAS), and let $\{v^*\} = V(N_1)\setminus V(N_2)$.

     Then for any $u,w \in V(N_2)$, it holds that $N_1$ has an SP-subgraph from $u$ to $w$ if and only if $N_2$ has an SP-subgraph from $u$ to $w$.
     In particular, if $Z$ is an SP-subgraph from $u$ to $w$ in $N_i$ for $i \in \{1,2\}$, there exists an SP-subgraph $Z'$ from $u$ to $w$ in $N_{3-i}$ with $V(Z)\Delta V(Z') \subseteq \{v^*\}$, and $Z'$ is degenerate if and only if $Z$ is.
     Furthermore $w$ has the same number of incoming arcs and incident edges in $Z$ as in $Z'$.
 \end{lemma}
 \begin{proof}

    If $N_2$ is derived from $N_1$ by an application of (PAS), then for some $u^*,w^*\in V(N_1)$ (with $u^*$ not the root of $N$), there exist parallel arcs $(u^*,v^*)$ and a single arc $(v^*,w^*)$ \mj{or edge $\{v^*,w^*\}$} (because $v^*$ has degree $3$). Then $N_1[\{u^*,v^*,w^*\}]$
    %{\color{blue} I am not sure what the [...] brackets mean?}
    % LEO: inserted definition above
    is an SP-subgraph from $u^*$ to $w^*$ in $N_1$, and $N_2$ is derived from $N_1$ by replacing this SP-subgraph with the arc $(u^*,w^*)$ \mj{or edge $\{u^*,w^*\}$}.
    On the other hand if $N_1$ is derived from $N_1$ by an application of (V1), (V2) or (V3), then $v^*$ has degree $2$ and neighbors $u^*, w^*$, and $N_1[\{u^*,v^*,w^*\}]$ is again an SP-subgraph from $u^*$ to $w^*$ in $N_1$, and $N_2$ is derived from $N_1$ by replacing this SP-subgraph with an arc or edge from $u^*$ to $w^*$.
    Thus, we may now assume that $Y_1^*: = N_1[\{u^*,v^*,w^*\}]$ is an SP-subgraph from $u^*$ to $w^*$ in $N_1$, and that $N_2$ is derived from $N_1$ by replacing $Y_1$ with an arc/edge $e^*$ from $u^*$ to $w^*$. (Note that this arc/edge itself also forms an SP-subgraph from $u^*$ to $w^*$.)
    Note also that $e^*$ is an arc if and only if $w^*$ has in incoming arc (as opposed to an incident edge) in $Y_1^*$.

     Now consider any SP-subgraph $Z_1$ from $u$ and $w$ in $N_1$, with $v^*\notin\{u,w\}$. If $v^*\notin V(Z_1)$ then $Z_2: = Z_1$ is also an SP-subgraph from $u$ to $w$ in $N_2$. Otherwise, $V(Z_1)$ contains not just $v^*$ but also $u^*$ and $w^*$ (otherwise $v^*$ is not on a path from $u$ to $w$). Thus $Y_1^*$ is an SP-subgraph from $u^*$ to $w^*$ in $Z_1$. Let $Z_2$ be derived from $Z_1$ by replacing $Y_1^*$ with 
     the arc/edge $e^*$.
     % $Y_2^*$. 
     Then $Z_2$ is a subgraph from $u$ to $w$ in $N_2$, and by construction $Z_2$ is an SP-graph with $V(Z_1)\Delta V(Z_2) \subseteq \{v^*\}$.

     Conversely, consider any SP-subgraph $Z_2$ from $u$ and $w$ in $N_2$. If $Z_2$ does not contain $e^*$, then $Z_1:= Z_2$ is also an SP-subgraph from $u$ to $w$ in $N_1$.
     Otherwise, let $Z_1$ be derived from $Z_2$ by replacing $e^*$ with $Y_1^*$. Then $Z_1$ is a subgraph from $u$ to $w$ in $N_1$, and by construction $Z_1$ is an SP-subgraph with $V(Z_1)\Delta V(Z_2) \subseteq \{v^*\}$.

        It remains to show that $Z_1$ is degenerate if and only if $Z_2$ is degenerate, for both constructions described above. 
        Note that $w$ has the same number of incoming arcs and incident edges in $Z_1$ as in $Z_2$.
        Indeed these arcs/edges are the same in both SP-subgraphs unless $w=w^*$, in which case the claim follows by comparing $Y_1^*$ with $e^*$.         
     Note also that (for both constructions), $u$ has the same degree in $Z_2$ as in $Z_1$, unless $u = u^*$ and rule (PAS) was applied, in which case $u$ is not the root.
     It follows that $Z_2$ is degenerate if and only if $Z_1$ is degenerate.      
 \end{proof}

% Recall that (PAS) is not applied in the case when there are parallel arcs $(u,v)$ but $u$ has degree $2$  - i.e. when $u$ is the root.
% This slightly complicates the intuition that every SP-subgraph in $N$ becomes an arc in $SUPPP(N)$. To account for this, we say that an SP-subgraph  $Z$ from $u$ to $w$ in $N$ is \emph{degenerate} 

% \mj{It remains to prove the following lemma, from which \Cref{lem:SuppWellDefinedBrief} immediately follows.}
\mj{We can now prove the following lemma, \leonn{which we will use to show} that exhaustively applying (PAS), (V1), (V2), (V3) in any order results in the same network:}

\begin{lemma}\label{lem:SuppWellDefined}
    Let $N_1 = N, N_2, \dots, N_m$ be a sequence of networks, $m \ge 2$, such that $N_{i+1}$ is derived from $N_i$ by an application of \mj{(V1) or (V2) or} (V3) or (PAS), for each $i \in \{1,\dots, m-1\}$, \mj{and such that (V1),(V2),}(V3) and (PAS) do not apply to $N_m$.
    Then 
    \begin{enumerate}
        \item  For each vertex $v$ of $N$, $v \in V(N)\setminus V(N_m)$ if and only if 
         $v$ is an internal vertex of some non-degenerate SP-subgraph in $N$; 
        
        \item For each $u,w \in V(N_m)$, there is a single arc $(u,w) \in E(N_m)$ if and only if there is a non-degenerate SP-subgraph from $u$ to $w$ in $N$ which ends in an arc.
        
        \item For each $u,w \in V(N_m)$, there is a single edge $(u,w) \in E(N_m)$ if and only if there is a non-degenerate SP-subgraph from $u$ to $w$ in $N$ which ends in an edge.

        \item For each $u,w \in V(N_m)$, there are parallel arcs $(u,w) \in E(N_m)$ if and only if there is a minimal degenerate SP-subgraph in $N$ that is a degenerate SP-subgraph from $u$ to $w$.
    \end{enumerate}
    %Thus $\supp{N}$ is well-defined for directed and semi-directed networks, proving Lemma~\ref{lem:SuppWellDefinedBrief}.
\end{lemma}
\begin{proof}

%{\color{blue} in this proof it's a bit tricky to see where statements (1), (2), (3) and (4) are proven - maybe say when each is shown? Also, where do we use these statements later on?}

    Note that, for any vertex~$v$ removed by an application of (PAS) or (V3) on some %directed
    \leonn{network} $N'$, $v$ is part of an SP-subgraph $Z$ from $u$ to $w$ in $N'$, where $u$ and $w$ are the parent and child of $v$ respectively.
    Furthermore $Z$ is non-degenerate (as we do not apply (PAS) when $u$ is the root).

    \leonn{To prove Statement 1, first} suppose that $v \in V(N)\setminus V(N_m)$, and let $i$ be the unique integer for which $v \in V(N_i)\setminus V(N_{i+1})$. Then $v$ was removed by an application of (PAS) or (V3) on $N_i$, and so $v$ is part of a non-degenerate SP-subgraph $Z_i$ from $u$ to $w$ in $N_i$.
    If $i > 1$, then by Lemma~\ref{lem:SPsubgraphs} 
    there exists a non-degenerate SP-subgraph $Z_{i-1}$ from $u$ to $w$ in $N_{i-1}$, with $V(Z_{i-1}) \subseteq V(Z_i)$.
    Thus $v$ is also an internal vertex of $Z_i$.
    Repeating this argument, we see that $v$ is an internal vertex of a non-degenerate SP-subgraph from $u$ to $w$ in $N_1 = N$, as required.
    %{\color{blue} as required for which statement?}

    Conversely, suppose that $v$ is an internal vertex of a non-degenerate SP-subgraph from $u$ to $w$ in $N_1$. 
    Note that %{\color{blue} in
    $N_m$ 
    has no non-degenerate SP-graphs except for those subgraphs consisting of a single arc, 
    as otherwise one of (PAS) or (V3) would apply to $N_m$.
    So there exists some largest $i \in \{1,\dots, m-1\}$ such that $v$ is an internal vertex of a non-degenerate SP-subgraph in $N_i$, but not in $N_{i+1}$.
    Let $Z_i$ be such a non-degenerate subgraph, and suppose $Z_i$ is from $u$ to $w$. Let $v^*$ be the unique vertex in $V(N_i)\setminus V(N_{i+})$. 
    Note that if $u = v^*$ then $v$ is also part of a non-degenerate SP-graph from $u^*$ to $w$ for $u^*$ the parent of $u$, and if $w = v^*$ then $v$ is part of a non-degenerate SP-subgraph from $u$ to $w^*$ for $w^*$ the child of $w$.
    %\todo{Proof needed?}
    Thus, we may assume neither $u$ nor $w$ is $v^*$.
    So Lemma~\ref{lem:SPsubgraphs} implies that $N_{i+1}$ also has a non-degenerate SP-graph from $u$ to $w$, with $V(Z_{i+1}) \supseteq V(Z_i)\setminus \{v^*\}$. As $v$ cannot be in $Z_{i+1}$ by choice of $i$, it follows that $v = v^*$, and so $v \notin V(N_m)$, as required.

    We now have that $v \in V(N_m)$ if and only if $v$ is not part of a non-degenerate SP-graph in $N$ \leonn{(Statement 1)}.
    It remains to consider the arcs and edges of $N_m$ \leonn{(Statements 2-4)}.

    By Lemma~\ref{lem:SPsubgraphs}, for any $u,w \in V(N_{i+1})$, there is a non-degenerate SP-subgraph $Z$ from $u$ to $w$ in $N_{i+1}$ if and only if there is a non-degenerate SP-subgraph from $u$ to $w$ in $N_i$, for all $i \in \{1,\dots, m-1\}$.
    It follows that there is a non-degenerate SP-subgraph from $u$ to $w$ in $N_m$ if and only if there is a non-degenerate SP-subgraph from $u$ to $w$ in $N_1 = N$.
    But the only non-degenerate SP-subgraphs in $N_m$ are arcs and edges. Moreover, the SP-subgraph from $u$ to $w$ ends in an incoming arc of $w$ if and only if $Z$ ends in an incoming arc of $w$.
    So $N_m$ has an arc from $u$ to $w$ if $Z$ ends in an arc, and $N_m$ has an edge between $u$ and $w$ if $Z$ ends in an edge. \leonn{This concludes the proof of Statements~2 and~3.}

    Finally, again by Lemma~\ref{lem:SPsubgraphs}, there is a degenerate SP-subgraph from $u$ to $w$ in $N_m$ if and only if there is a degenerate SP-subgraph from $u$ to $w$ in $N_1 = N$. \leonn{Statement 4 now follows since,} in $N_m$ the only degenerate \leonn{SP-subgraphs are pairs} of parallel arcs.
\end{proof}

%\mj{Finally, we can prove \Cref{lem:SuppWellDefinedBrief} which by the last lemma only requires us to consider the (BLS) operation.}

\textbf{\Cref{lem:SuppWellDefinedBrief}.}
    \emph{$\supp{N}$ is well-defined \leonn{for any network~$N$}.
}

\begin{proof} %[Proof of \Cref{lem:SuppWellDefinedBrief}]
   \mj{Let $N_1$ be the network derived from $N$ by applying a %sequence of
   \leonn{(BLS) operation} to every 
   % degree-2
   blob with at most two incident edge/arcs in $N$. Note that suppressing one blob does not affect the other blobs in the network, and so $N_1$ \leonn{is} %be
   well-defined. 
    %By
    \leonn{Considering the} definition of $\supp{N}$, it remains to show that starting with $N_1$ and exhaustively applying the operations (PAS), (V1), (V2), (V3) will always result in the same network.}

%{\color{blue} I am a bit confused in the next para -- it seems we are not assuming $m=m'$? So should we use the same index $i$ for the two sequences? Or does it follow that $m=m'$?
% LEO: we do not assume m=m'. It does follow from the lemma that m=m' since each of these operations reduces the number of vertices by exactly 1. This is not relevant for the proof so I don't think we should mention it.
% Also, why do we need this extra argument, given the statement of the last lemma?
% LEO: Rephrased.

    \mj{To see this, 
    let $N_1, N, N_2, \dots, N_m$ and $N'_1 = N_1, N, N'_2, \dots, N'_{m'}$ be two sequences of networks, such that $N_{i+1}$ (respectively, $N'_{i+1}$) is derived from $N_i$ ($N'_i)$ by an application of \mj{(V1) or (V2) or} (V3) or (PAS), for each $i \in \{1,\dots, m-1\}$ ($i \in \{1,\dots, m'-1\}$),
    and such that (V1),(V2),(V3) and (PAS) do not apply to $N_m$ ($N'_{m'}$).
    By applying~\Cref{lem:SuppWellDefined} to $N_m$ and $N'_{m'}$, we see that $N_m$ and $N'_{m'}$ have exactly the same vertices \leonn{(Statement~1 of~\Cref{lem:SuppWellDefined})}, arcs \leonn{(Statement~2 of~\Cref{lem:SuppWellDefined})}, edges \leonn{(Statement~3 of~\Cref{lem:SuppWellDefined})} and parallel arcs \leonn{(Statement~4 of~\Cref{lem:SuppWellDefined})}. Thus, $N_m$ and $N'_{m'}$ are the same network, and so $\supp{N}$ is well-defined.}
\end{proof}

\begin{figure}
    \centerline{\includegraphics{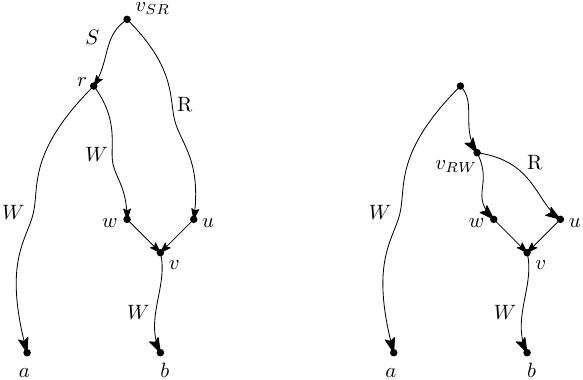}}
    \caption{\label{fig:wedges} Illustrations for the proof of Lemma~\ref{lem:wedges} for the case that~$N$ is directed.}
\end{figure}

\textbf{Lemma~\ref{lem:wedges}.}
    \emph{Consider a network~$N$ on~$X$, leaves~$a,b\in X$ and a reticulation~$v$ with parents~$u,w$. If~$v$ is on a $\wedge$-path in~$N$ between~$a$ and~$b$, then~$u$ is on a $\wedge$-path in~$N$ between~$a$ and~$b$.}
\begin{proof}
    First suppose that~$N$ is a directed network.

    Consider any $\wedge$-path~$W$ between~$a$ and~$b$ containing~$v$. It contains at least one of~$u$ and~$w$. If~$W$ contains~$u$ then the lemma holds. Hence, suppose that~$W$ does not contain~$u$ and hence traverses the arc~$(w,v)$. Assume without loss of generality that arc~$(w,v)$ is traversed on the part of~$W$ directed towards~$b$. Consider any directed path~$R$ from the root of~$N$ to~$u$ (which exists as $N$ is assumed to be directed).
    
    First suppose that~$R$ is disjoint from~$W$. Let~$r$ be the vertex of~$W$ such that~$W$ consists of directed paths from~$r$ to~$a$ and~$b$. Then consider a directed path~$S$ from the root to~$r$. Let~$v_{SR}$ be the last common vertex of~$S$ and~$R$. Then a $\wedge$-path between~$a$ and~$b$ containing~$u$ can be obtained by following~$W$ from~$a$ to~$r$, then following~$S$ to~$v_{SR}$, following~$R$ to~$u$, following the arc $(u,v)$, and finally following~$W$ from~$v$ to~$b$. See Figure~\ref{fig:wedges} (left) for an example.

    Now suppose~$R$ intersects~$W$. Let~$v_{RW}$ be the last vertex of~$R$ that is on~$W$. Then
    %{\color{blue} without loss of generality? (since you might have to follow $W$ from $b$ to $v_{RW}$?)}
    % LEO: we already assume wlog above that b is below (w,v)
    a $\wedge$-path between~$a$ and~$b$ containing~$u$ can be obtained by following~$W$ from~$a$ to~$v_{RW}$, then following~$R$ to~$u$, following the arc $(u,v)$, and finally following~$W$ from~$v$ to~$b$. See Figure~\ref{fig:wedges} (right) for an example.

    It remains to consider the case that~$N$ is semi-directed. Consider any rooting~$D$ of~$N$. If~$D$ contains arc~$(u,v)$, then~$u$ is on a $\wedge$-path in~$D$ between~$a$ and~$b$ by the first part of the proof (for directed networks). Hence,~$u$ is on a $\wedge$-path in~$N$ between~$a$ and~$b$.

    Otherwise,~$D$ contains arcs~$(\rho,u),(\rho,v)$. Then,~$\rho$ is on a $\wedge$-path in~$D$ between~$a$ and~$b$ by the first part of the proof. Hence,~$u$ is on a $\wedge$-path in~$N$ between~$a$ and~$b$.
\end{proof}

\end{document}